\documentclass[11pt,letterpaper]{article}
\pdfoutput=1

\usepackage[final=true,colorlinks=true,citecolor=blue]{hyperref}

\usepackage[rgb]{xcolor}

\usepackage{todonotes}

\usepackage{amsmath,amsthm,amssymb,amsfonts}
\usepackage[left=1in,right=1in,top=1in,bottom=1in]{geometry}
\usepackage{authblk}

\usepackage{cite}
\usepackage[outline]{contour}
\usepackage{amssymb,amsmath}
\usepackage{braket}
\usepackage{algorithm}
\usepackage{graphicx}
\usepackage[noend]{algpseudocode}
\usepackage{mathtools}

\usepackage{tikz}
\usetikzlibrary{patterns,calc,arrows,arrows.meta,positioning,fit,shapes}

\usepackage[font=small]{caption} 
\usepackage{subcaption}

\usepackage{thmtools}
\usepackage{thm-restate}

\makeatletter
\def\input@path{{figures/}}
\makeatother

\usepackage[capitalize]{cleveref}

\newtheorem{theorem}{Theorem}[section]
\newtheorem{proposition}{Proposition}
\newtheorem{observation}{Observation}[section]
\newtheorem{lemma}{Lemma}
\newtheorem{corollary}{Corollary}
\newtheorem{definition}{Definition}[section]

\newtheorem{question}{Question}[section]

\newcommand{\Oh}{\mathcal{O}}
\newcommand{\tOh}{\widetilde{\tOh}}

\newcommand{\D}{{\bf D}}

\newcommand{\floor}[1]{\left\lfloor#1\right\rfloor}
\newcommand{\ceil}[1]{\left\lceil#1\right\rceil}
\newcommand{\head}{\textsf{head}}
\newcommand{\tail}{\textsf{tail}}
\newcommand{\Tphrase}{\mathcal{T}_{\mathrm{phrase}}}
\newcommand{\Sparse}{\mathcal{S}_{\mathrm{parse}}}
\newcommand{\Tparse}{T_{\mathrm{parse}}}
\newcommand{\lce}{\textnormal{\textsc{lce}}}

\newcommand{\quotes}[1]{``#1''}
\newcommand{\absolute}[1]{\left\lvert#1\right\rvert}
\newcommand{\tuple}[1]{\left\langle#1\right\rangle}

\usepackage{enumitem}
\setlist[description]{leftmargin=\parindent,labelindent=\parindent}

\makeatletter
\newcommand\thankssymb[1]{\textsuperscript{\@fnsymbol{#1}}}
\makeatother

\makeatletter

\begin{document}

\title{Optimal Square Detection Over General Alphabets}

\author[1]{Jonas Ellert\thanks{Partially supported by the German Research Association (DFG) within the Collaborative Research Center SFB 876, project A6.}}
\author[2]{Paweł Gawrychowski\thanks{Partially supported by the grant ANR-20-CE48-0001 from the French National Research Agency (ANR).}}
\author[3]{Garance Gourdel\thankssymb{2}}

\affil[1]{Technical University of Dortmund, Germany}
\affil[2]{Institute of Computer Science, University of Wrocław, Poland}
\affil[3]{DI/ENS, PSL Research University, IRISA Inria Rennes, France}

\date{}

\maketitle

\begin{abstract}
Squares (fragments of the form $xx$, for some string $x$) are arguably the most natural type of repetition in strings.
The basic algorithmic question concerning squares is to check if a given string of length $n$ is square-free, that is, does not contain
a fragment of such form. Main and Lorentz [J. Algorithms 1984] designed an $\mathcal{O}(n\log n)$ time algorithm
for this problem, and proved a matching lower bound assuming the so-called general alphabet, meaning that the algorithm
is only allowed to check if two characters are equal. 
However, their lower bound also assumes that there are $\Omega(n)$ distinct symbols in the string.
As an open question, they asked if there is a faster algorithm if one restricts
the size of the alphabet. Crochemore [Theor. Comput. Sci. 1986] designed a linear-time algorithm for constant-size alphabets,
and combined with more recent results his approach in fact implies such an algorithm for linearly-sortable alphabets.
Very recently, Ellert and Fischer [ICALP 2021] significantly relaxed this assumption by designing a linear-time algorithm for
general ordered alphabets, that is, assuming a linear order on the characters that permits constant time order comparisons. However, the open question of Main and Lorentz
from 1984 remained unresolved for general (unordered) alphabets.  In this paper, we show that testing square-freeness of a length-$n$ string over
general alphabet of size $\sigma$ can be done with $\mathcal{O}(n\log \sigma)$ comparisons, and cannot be done with $o(n\log \sigma)$ comparisons.
We complement this result with an $\mathcal{O}(n\log \sigma)$ time algorithm in the Word RAM model.
Finally, we extend the algorithm to reporting all the runs (maximal repetitions) in the same complexity.
\end{abstract}

\section{Introduction}

The notion of repetition is a central concept in combinatorics on words and algorithms on strings. In this context,
a word or a string is simply a sequence of characters from some finite alphabet $\Sigma$. In the most basic version,
a repetition consists of two (or more) consecutive occurrences of the same fragment. Repetitions are interesting not
only from a purely theoretical point of view, but are also very relevant in bioinformatics~\cite{Kolpakov2003}.
A repetition could be a square, defined as two consecutive occurrences of the same fragment, a higher power (for example, a cube), or a run, which is a length-wise maximal periodic substring.
For example, both \texttt{anan} and \texttt{nana} are squares with two occurrences each in \texttt{banananas}, and they belong to the same run \texttt{ananana}.
In this paper, we start by focusing on squares, then generalize our results for runs.

The study of squares in strings goes back
to the work of Thue published in 1906~\cite{thue1906}, who considered the question of constructing an infinite word
with no squares. It is easy to see that any sufficiently long binary word must contain a square, and Thue proved that
there exists an infinite ternary word with no squares. His result has been rediscovered multiple times, and in 1979
Bean, Ehrenfeucht and McNulty~\cite{BeanEM1979} started a systematic study of the so-called avoidable repetitions,
see for example the survey by Currie~\cite{Currie05}.

\paragraph{Combinatorics on words.} The basic tool in the area of combinatorics on words is the so-called
periodicity lemma. A period of a string $T[1..n]$ is an integer $d$ such that $T[i]=T[i+d]$ for every $i\in [1,n-d]$,
and the periodicity lemma states that if $p$ and $q$ are both such periods and $p+q\leq n+\gcd(p,q)$ then 
$\gcd(p,q)$ is also a period~\cite{Fine1965}. 
This was generalised in a myriad of ways, for strings~\cite{Castelli1999,Justin2000,Tijdeman2003},
partial words (words with don't cares)~\cite{Berstel1999,Blanchet-Sadri2008,Blanchet-Sadri2002,Shur2004,Shur2001,Idiatulina2014,Kociumaka2022},
Abelian periods\cite{Constantinescu2006,Blanchet-Sadri2013}, parametrized periods~\cite{Apostolico2008},
order-preserving periods~\cite{Matsuoka2016,Gourdel2020}, approximate periods~\cite{Amir2010,Amir2012,Amir2015}.
Now, a square can be defined as a fragment of length twice its period. 
The string $\texttt{a}^{n}$ contains $\Omega(n^{2})$ such fragments,
thus from the combinatorial point of view it is natural to count only distinct squares.
Fraenkel and Simpson~\cite{Fraenkel1998} showed an upper bound of $2n$ and a lower bound of $n-\Theta(\sqrt{n})$ for the maximum number of distinct squares in a length-$n$ string.
After a sequence of improvements~\cite{Ilie2007,Deza2015,Thierry2020}, the upper bound was very recently improved to $n$~\cite{Brlek2022}.
The last result was already generalised to higher powers~\cite{Li2022}.
Another way to avoid the trivial examples such as $\texttt{a}^{n}$ is to count only maximal periodic fragments,
that is, fragments with period at most half of their length and that cannot be extended to the left or to the right without
breaking the period. Such fragments are usually called runs. Kolpakov and Kucherov~\cite{Kolpakov1999} showed
an upper bound of $\Oh(n)$ on their number, and this started a long line of work on determining the exact
constant~\cite{Rytter2006,Puglisi2008,Crochemore2008,Giraud2008,Giraud2009,Crochemore2011}, culminating
in the paper of Bannai et al.~\cite{Bannai2017} showing an upper bound of $n$, and followed by even better upper bounds
for binary strings~\cite{Fischer2015,Holub2017}. This was complemented by a sequence of 
lower bounds~\cite{Franek2008,Matsubara2008,Matsubara2009,Simpson2010}.

\paragraph{Algorithms on strings. }
In this paper, we are interested in the algorithmic aspects of detecting repetitions in strings. The most basic question
in this direction is checking if a given length-$n$ string contains at least one square,
while the most general version asks for computing all the runs.
Testing square-freeness was
first considered by Main and Lorentz~\cite{Main1984}, who designed an $\Oh(n\log n)$ time algorithm based on
a divide-and-conquer approach and a linear-time procedure for finding all new squares obtained when concatenating
two strings. In fact, their algorithm can be used to find (a compact representation of) all squares in a given string
within the same time complexity. They also proved that any algorithm based on comparisons of characters needs
$\Omega(n\log n)$ such operations to test square-freeness in the worst case. Here, comparisons of characters means
checking if characters at two positions of the input string are equal. However, to obtain the lower bound they
had to consider instances consisting of even up to $n$ distinct characters, that is, over alphabet of size $n$.
This is somewhat unsatisfactory, and motivates the following open question that was explicitly asked by Main and Lorentz~\cite{Main1984}:

\begin{question}
Is there a faster algorithm to determine if a string is square-free if we restrict the size of the alphabet?
\end{question}

Another $\Oh(n\log n)$ time algorithm for finding all repetitions was given by Crochemore~\cite{Crochemore1981},
who also showed that for constant-size alphabets testing square-freeness can be done in  $\Oh(n)$ time~\cite{Crochemore1986}.
In fact, the latter algorithm works in $\Oh(n\log \sigma)$ time for alphabets of size $\sigma$ with a linear order on the characters.
That is, it needs to test if the character at some position is smaller than the character at another position.
In the remaining part of the paper, we will refer to this model as general ordered alphabet, while the model
in which we can only test equality of characters will be called general (unordered) alphabet.
Later, Kosaraju~\cite{Kosaraju1994} showed that in fact, assuming constant-size alphabet, $\Oh(n)$ time is enough
to find the shortest square starting at each position of the input string.
Apostolico and Preparata~\cite{Apostolico1983} provide another $\Oh(n\log n)$ time algorithm assuming a general ordered alphabet,
based more on data structure considerations than combinatorial properties of words.
Finally, a number of alternative $\Oh(n\log n)$ and $\Oh(n\log \sigma)$ time algorithms (respectively, for general unordered
and general ordered alphabets) can be obtained from the work on online~\cite{Hong2008,Kosolobov2014,Kosolobov2015a}
and parallel~\cite{Apostolico1996} square detection (interestingly, this cannot be done efficiently in the related
streaming model~\cite{Merkurev2019,Merkurev2022}).

Faster algorithms for testing square-freeness of strings over general ordered alphabets  were obtained as a byproduct of
the more general results on finding all runs. Kolpakov and Kucherov~\cite{Kolpakov1999} not only proved that any length-$n$
string contains only $\Oh(n)$ runs, but also showed how to find them in the same time assuming
linearly-sortable alphabet. Every square is contained in a run, and every run contains at least one square, thus this
in particular implies a linear-time algorithm for testing square-freeness over such alphabets. For general ordered alphabets,
Kosolobov~\cite{Kosolobov2015} showed that the decision tree complexity of this problem is only $\Oh(n)$, and later complemented this with an efficient $\Oh(n(\log n)^{2/3})$ time algorithm~\cite{Kosolobov2016}
(still using only $\Oh(n)$ comparisons). The time complexity was then improved to $\Oh(n\log\log n)$ by providing a general
mechanism for answering longest common extension (LCE) queries for general ordered alphabets~\cite{Gawrychowski2016},
and next to $\Oh(n\alpha(n))$ by observing that the LCE queries have additional structure~\cite{CrochemoreIKKPR16}.
Finally, Ellert and Fischer provided an elegant $\Oh(n)$ time algorithm, thus fully resolving the complexity of square detection
for general ordered alphabets. However, for general (unordered) alphabets the question of Main and Lorentz remains
unresolved, with the best upper bound being $\Oh(n\log n)$, and only known to be asymptotically tight for alphabets of
size $\Theta(n)$.

\paragraph{General alphabets.} While in many applications one can without losing generality assume some
ordering on the characters of the alphabet, no such ordering is necessary for defining what a square is. Thus, it is natural
from the mathematical point of view to seek algorithms that do not require such an ordering to efficiently test square-freeness.
Similar considerations have lead to multiple beautiful results concerning the pattern matching problem, such as constant-space algorithms~\cite{Galil1983,Breslauer1992},
or the works on the exact number of required equality comparisons ~\cite{Cole1995,Cole1997} 
 More recent examples include the work of Duval, Lecroq, and Lefebvre~\cite{Duval2014}
on computing the unbordered conjugate/rotation, and Kosolobov~\cite{Kosolobov2016a} on finding the leftmost critical point.

\paragraph{Main results.}

We consider the complexity of checking if a given string $T[1..n]$ containing $\sigma$ distinct characters is square-free. 
The input string can be only accessed by issuing comparisons $T[i]\stackrel{?}{=} T[j]$, and the value of $\sigma$ is not
assumed to be known. We start by analysing the decision tree complexity of the problem. That is, we only
consider the required and necessary number of comparisons, without worrying about an efficient implementation.
We show that, even if the value of $\sigma$ is assumed to be known, $\Omega(n\log \sigma)$ comparisons are required. 

\begin{restatable}{theorem}{lowerbound}
\label{thm:lowerbound}
For any integers $n$ and $\sigma$ with $8 \leq \sigma \leq n$, there is no deterministic algorithm that performs at most $n \ln \sigma - 3.6n = \Oh(n \ln \sigma)$ comparisons in the worst case, and determines whether a length-$n$ string with at most $\sigma$ distinct symbols from a general unordered alphabet is square-free.
\end{restatable}

Next, we show that $\Oh(n\log \sigma)$ comparisons are sufficient. We stress that the value of $\sigma$ is not assumed to
be known. In fact, as a warm-up for the above theorem, we first prove that finding a sublinear multiplicative approximation
of this value requires $\Omega(n\sigma)$ comparisons. This does not contradict the claimed upper bound, as we are only saying
that the number of comparisons used on a particular input string is at most $\Oh(n\log \sigma)$, but might actually be smaller.
Thus, it is not possible to extract any meaningful approximation of the value of $\sigma$ from the number of used comparisons.

\begin{restatable}{theorem}{upperbound}
\label{thm:upperbound}
Testing square-freenes of a length-$n$ string that contains $\sigma$ distinct symbols from a general unordered alphabet can be done with $\Oh(n \log \sigma)$ comparisons.
\end{restatable}

The proof of the above result is not efficient in the sense that it only restricts the overall number of comparisons, and not the time
to actually figure out which comparisons should be used.  A direct implementation results in a quadratic time algorithm. We first
show how to improve this to $\Oh(n\log \sigma+n\log^{*}n)$ time (while still keeping the asymptotically optimal $\Oh(n\log \sigma)$ number
of comparisons), and finally to $\Oh(n\log \sigma)$. In this part of the paper, we assume the Word RAM model with word of length $\Omega(\log n)$.
We stress that the input string is still assumed to consist of characters that can be only tested for equality, that is, one should
think that we are given oracle access to a functions that, given $i$ and $j$, checks whether $T[i]=T[j]$.

\begin{restatable}{theorem}{upperbound2}
\label{thm:upperbound2}
Testing square-freeness of a length-$n$ string that contains $\sigma$ distinct symbols from a general unordered alphabet can be implemented in $\Oh(n\log \sigma)$ comparisons
and time.
\end{restatable}

Finally, we also generalize this result to the computation of runs.

\begin{restatable}{theorem}{upperbound2runs}
\label{thm:upperbound:runs}
Computing all runs in a length-$n$ string that contains $\sigma$ distinct symbols from a general unordered alphabet can be implemented in $\Oh(n\log \sigma)$ comparisons
and time.
\end{restatable}

Altogether, our results fully resolve the open question of Main and Lorentz for the case of general unordered alphabets and deterministic algorithms.
We leave extending our lowerbound to randomised algorithms as an open question.

\paragraph{Overview of the methods.}

As mentioned before, Main and Lorentz~\cite{Main1984} designed an $\Oh(n\log n)$ time algorithm for testing square-freeness of
length-$n$ strings over general alphabets. The high-level idea of their algorithm goes as follows. They first designed a procedure for checking,
given two strings $x$ and $y$, if their concatenation contains a square that is not fully contained in $x$ nor $y$ in $\Oh(\absolute{x}+\absolute{y})$ time.
Then, a divide-and-conquer approach can be used to detect a square in the whole input string in $\Oh(n\log n)$ total time.
For general alphabets of unbounded size this cannot be improved, but Crochemore~\cite{Crochemore1986} showed that, for general
ordered alphabets of size $\sigma$, a faster $\Oh(n\log \sigma)$ time algorithm exists. The gist of his approach is to first obtain the so-called
$f$-factorisation of the input string (related to the well-known Lempel-Ziv factorisation), that in a certain sense ``discovers'' repetitive
fragments. Then, this factorisation can be used to apply the procedure of Main and Lorentz on appropriately selected fragments of the input
strings in such a way that the leftmost occurrence of every distinct square is detected, and the total length of the strings on which we apply the
procedure is only $\Oh(n)$. The factorisation can be found in $\Oh(n\log \sigma)$ time for general ordered alphabets of size $\sigma$ by,
roughly speaking, constructing some kind of suffix structure (suffix array, suffix tree or suffix automaton).

For general (unordered)
alphabets, computing the $f$-factorisation (or anything similar) seems problematic, and in fact we show (as a corollary of our lower
bound on approximating the alphabet size) that computing the $f$-factorisation or Lempel-Ziv-factorisation (LZ-factorisation) of a given length-$n$
string containing $\sigma$ distinct characters requires $\Omega(n\sigma)$ equality tests. Thus, we need another approach.
Additionally, the $\Oh(n)$ time algorithm of Ellert and Fischer~\cite{Ellert2021} hinges on the notion of Lyndon words, which is simply
not defined for strings over general alphabets. Thus, at first glance it might seem that $\Theta(n\sigma)$ is the right time complexity
for testing square-freeness over length-$n$ strings over general alphabets of size $\sigma$. However, due to the $\Omega(n\log n)$
lower bound of Main and Lorentz for testing square-freeness of length-$n$ string consisting of up to $n$ distinct characters,
one might hope for an $\Oh(n\log \sigma)$ time algorithm when there are only $\sigma$ distinct characters.

We begin our paper with a lower bound of $\Theta(n\log \sigma)$ for such strings. Intuitively, we show that testing square-freeness
has the direct sum property: $\frac{n}{\sigma}$ instances over length-$\sigma$ strings can be combined into a single instance over length-$n$ string.
As in the proof of Main and Lorentz, we use the adversarial method. While the underlying calculation is essentially the same,
we need to appropriately combine the smaller instances, which is done using the infinite square-free Prouhet-Thue-Morse sequence, and use significantly
more complex rules for resolving the subsequent equality tests. As a warm-up for the adversarial method, we prove that computing
any meaningful approximation of the number of distinct characters requires $\Omega(n\sigma)$ such tests, and that this implies
the same lower bound on computing the $f$-factorisation and the Lempel-Ziv factorisation (if the size of the alphabet is unknown in advance).

We then move to designing an approach that uses $\Oh(n\log \sigma)$ equality comparisons to test square-freeness.
As discussed earlier, one way of detecting squares uses the $f$-factorisation of the string, which is similar to its LZ factorisation.
However, as we prove in Corollary~\ref{cor:f-facto} and \ref{cor:LZ}, we cannot compute either of these factorisations over a general unordered alphabet in $o(n\sigma)$ comparisons.
Therefore, we will instead use a novel type of factorisation, $\Delta$-approximate LZ factorisation, that can be seen as an approximate
version of the LZ factorisation.
Intuitively, its goal is to ``capture'' all sufficiently long squares, while the original LZ factorisation (or $f$-factorisation)
captures all squares. Each phrase in a $\Delta$-approximate LZ factorisation consists of a head of length at most $\Delta$ and a tail
(possibly empty) that must occur at least once before, such that the whole phrase is at least as long as the classical LZ phrase starting
at the same position. Contrary to the classical LZ factorisation, this factorisation is not unique. 
The advantage of our modification is that there are fewer phrases (and there is more flexibility as to what they should be), and
hence one can hope to compute such factorisation more efficiently.

To design an efficient construction method for $\Delta$-approximate LZ factorisation, we first show how to compute a sparse
suffix tree while trying to use only a few symbol comparisons. This is then applied on a set of positions from a so-called
difference cover with some convenient synchronizing properties.
Then, a $\Delta$-approximate LZ factorisation allows us to detect squares of length $\geq 8\Delta$.

The first warm-up algorithm fixes $\Delta$ depending on $n$ and $\sigma$ (assuming that $\sigma$ is known), and uses
the approximate LZ factorisation to find all squares of length at least $8\Delta$. It then finds all the shorter squares by dividing
the string in blocks of length $8\Delta$, and applying the original algorithm by Main and Lorentz on each
block pair. Our choice of $\Delta$ leads to $\Oh(n (\lg \sigma + \lg \lg n))$ comparisons.

The improved algorithm does not need to know $\sigma$, and instead starts with a large $\Delta = \Omega(n)$, and then
progressively decreases $\Delta$ in at most $\Oh(\lg \lg n)$ phases, where later phases detect shorter squares.
As soon as we notice that there are many distinct characters in the alphabet, by carefully adjusting the parameters
we can afford switching to the approach of Main and Lorentz on sufficiently short fragments of the input string. 
Since we cannot afford $\Omega(n)$ comparisons per phase, we use a deactivation technique, where whenever we perform a large
number of comparisons in a phase, we will discard a large part of the string in all following phases.
More precisely, during a given phase, we avoid looking for squares in a fragment fully contained in a tail from an earlier phase.
This leads to optimal $\Oh(n \lg \sigma)$ comparisons.

The above approach uses an asymptotically optimal number of equality tests in the worst case, but does not result in an efficient
algorithm. The main bottleneck is constructing the sparse suffix trees. However, it is not hard to provide an efficient implementation
using the general mechanism for answering LCE queries for strings over general alphabets~\cite{Gawrychowski2016}.
Unfortunately, the best known approach for answering such queries incurs an additional $\Oh(n\log^{*}n)$ in the time complexity,
even if the size of the alphabet is constant. We overcome this technical hurdle by carefully deactivating fragments
of the text to account for the performed work.

Many of our techniques can easily be modified to compute all runs rather than detecting squares. We exploit that the approximate
factorisation reveals long substrings with an earlier occurrence. Hence we compute runs only for the first occurrence of such substrings,
while for later occurrences we simply copy the already computed runs.
By carefully arranging the order of the computation, we ensure that the total time for copying is bounded by the number
of runs, which is known to be $\Oh(n)$.
This way, we achieve $\Oh(n \lg \sigma)$ time and comparisons to compute all runs.

\section{Preliminaries}
\label{sec:prelim}

\paragraph{Strings.}
A string of length $n$ is a sequence $T[1] \dots T[n]$ of characters from a finite alphabet $\Sigma$ of size $\sigma$. The substring $T[i..j]$ is the string $T[i] \cdots T[j]$, whereas the fragment $T[i..j]$ refers to the specific occurrence of $T[i..j]$ starting at position $i$ in $T$. If $i > j$, then $T[i..j]$ is the empty string.
A suffix of $T$ has the form $T[i..n]$.
We say that a fragment $T[i'..j']$ is properly contained in another fragment $T[i..j]$ if $i < i' \leq j' < j$.
A substring is properly contained in $T[i..j]$, if it equals a fragment that is properly contained in $T[i..j]$.
We write $T[i..j)$ a shortcut for $T[i..j-1]$. Similarly, we write $[i, j] = [i, j + 1)$ as a shortcut for the integer interval $\{i, \dots, j\}$. Given two positions $i \leq j$, their longest common extension (LCE) is the length of the longest common prefix between suffixes $T[i..n]$ and $T[j..n]$, formally defined as $\lce(i,j) = \lce(j,i) = \max\{\ell \in \{0, \dots, n - j + 1\} \mid T[i..i+\ell) = T[j..j+\ell)\}.$

\begin{definition}
A positive integer $p$ is a period of a string $T[1 .. n]$ if $T[i] = T[i+p]$ for every $i\in \{1, \dots n-p \}$. The smallest such $p$ is called the period of $T[1..n]$, and we call a string periodic if its period $p$ is at most $\frac{n}{2}$.
\end{definition}

\paragraph{Computational model.}
For a general unordered alphabet $\Sigma$, the only allowed operation on the characters is comparing for equality. In particular, there is no
linear order on the alphabet. Unless explicitly stated otherwise, we will only use such comparisons. A general ordered alphabet has a total order, such that comparisons of the type less-equals are possible. 

In the algorithmic part of the paper, we assume the standard unit-cost Word RAM model with words of length $\Omega(\log n)$,
but the algorithm is only allowed to access the input string $T[1..n]$ by comparisons $T[i]\stackrel{?}{=}T[j]$, which are assumed to take constant time.
We say that a string of length $n$ is over a linearly-sortable alphabet, if we can sort the $n$ symbols of the string in $\Oh(n)$ time. Note that whether or not an alphabet is linearly-sortable depends not only on the alphabet, but also on the string. For example, the alphabet $\Sigma=\{1,\dots,m^{\Oh(1)}\}$ is linearly-sortable for strings of length $n = \Omega(m)$ (e.g., using radix sort), but it is unknown whether it is linearly-sortable for all strings of length $n = o(m)$~\cite{HanT02}.
Our algorithm will internally use strings over linearly-sortable alphabets. We stress that in such strings the characters are not the
characters from the input string, but simply integers calculated by the algorithm.
Note that every linearly-sortable alphabet is also a general ordered alphabet.

\paragraph{Squares and runs.}
A square is a length-$2\ell$ fragment of period $\ell$. 
The following theorem is a classical result by Main and Lorentz~\cite{Main1984}.

\begin{theorem}
\label{lem:classical}
Testing square-freenes of $T[1..n]$ over a general alphabet can be implemented in $\Oh(n\log n)$ time and comparisons.
\end{theorem}

\noindent The proof of the above theorem is based on running a divide-and-conquer procedure
using the following lemma.

\begin{lemma}
\label{lem:conquer}
Given two strings $x$ and $y$ over a general alphabet, we can test if there is a square in $xy$ that is not fully contained in $x$ nor $y$ in $\Oh(|x|+|y|)$
time and comparisons.
\end{lemma}

A repetition is a length-$\ell$ fragment of period at most $\frac \ell 2$. 
A run is a maximal repetition. 
Formally, a repetition in $T[1..n]$ is a triple $\tuple{s, e, p}$ with $s,e \in [1, n]$ and $p \in [1, \frac{e - s + 1}{2}]$ such that $p$ is the smallest period of $T[s..e]$. 
A run is a repetition $\tuple{s, e, p}$ that cannot be extended to the left nor to the right with the same period, in other words
$s = 1$ or $T[s-1]\neq T[s-1+p]$ and $e = n$ or $T[e+1] \neq T[e+1-p]$.
The celebrated runs conjecture, proven by Bannai et al.~\cite{Bannai2017}, states
that the number of runs is any length-$n$ string is less than $n$.
Ellert and Fischer~\cite{Ellert2021} showed that all runs in a string over a general ordered alphabet can be computed in $\Oh(n)$ time.
As mentioned earlier, each run contains a square, and each square is contained in a run.
Thus, the string contains a square if and only if it contains a run, and it follows:
\begin{theorem}
\label{lem:fasterclassical}
Computing all runs (and thus testing square-freeness) of $T[1..n]$ over a general ordered alphabet can be implemented in $\Oh(n)$ time.
\end{theorem}

\paragraph{Lempel-Ziv factorisation.}
The unique LZ phrase starting at position $s$ of $T[1..n]$ is a fragment $T[s..e]$ such that $T[s..(e-1)]$ occurs at least twice in $T[1..(e-1)]$ and either $e = n$ or $T[s..e]$ occurs only once in $T[1..e]$.
The Lempel-Ziv factorisation of $T$ consists of $z$ phrases $f_{1},\ldots,f_{z}$ such that the concatenation $f_{1}\ldots f_{z}$ is equal to $T[1..n]$ and each $f_{i}$ is the unique LZ phrase starting at position $1 + \sum_{j=1}^{i - 1}\absolute{f_{j}}$.

\paragraph{Tries.}
Given a collection $\mathcal S = \{T_1, \dots, T_k\}$ of strings over some alphabet $\Sigma$, its trie is a rooted tree with edge labels from $\Sigma$. 
For any node $v$, the concatenation of the edge labels from the root to the node \emph{spells} a string. 
The string-depth of a node is the length of the string that it spells. 
No two nodes spell the same string, i.e., for any node, the labels of the edges to its children are pairwise distinct. Each leaf spells one of the $T_i$, and each $T_i$ is spelled by either an internal node or a leaf.

The compacted trie of $\mathcal S$ can be obtained from its (non-compacted) trie by contracting each path between a leaf or a branching node and its closest branching ancestor into a single edge (i.e., by contraction we eliminate all non-branching internal nodes). The label of the new edge is the concatenation of the edge labels of the contracted path in root to leaf direction. Since there are at most $k$ leaves and all internal nodes are branching, there are $O(k)$ nodes in the compacted trie. Each edge label is some substring $T_i[s..e]$ of the string collection, and we can avoid explicitly storing the label by instead storing the reference $(i,s,e)$.  Thus $O(k)$ words are sufficient for storing the compacted trie. Consider a string $T'$ that is spelled by a node of the non-compacted trie. We say that $T'$ is explicit, if and only if it is spelled by a node of the compacted trie. 
Otherwise $T'$ is implicit.

The suffix tree of a string $T[1..n]$ is the compacted trie containing exactly its suffixes, i.e., a trie over the string collection $\{T[i..n] \mid i \in \{1, \dots, n\}\}$. It is one of the most fundamental data structures in string algorithmics, and is widely used, e.g., for compression and indexing~\cite{Gusfield1997}. The suffix tree can be stored in $\Oh(n)$ words of memory, and for linearly-sortable alphabets it can be computed in $\Oh(n)$ time~\cite{Farach1997}.
The sparse suffix tree of $T$ for some set $B \subseteq \{1, \dots, n\}$ of sample positions is the compacted trie containing exactly the suffixes $\{T[i..n] \mid i \in B\}$. It can be stored in $\Oh(\absolute{B})$ words of memory.

We assume that $T$ is terminated by some special symbol $T[n] = \texttt\textdollar$ that occurs nowhere else in $T$. This ensures that each suffix is spelled by a leaf, and we label the leaves with the respective starting positions of the suffixes. Note that for any two leaves $i\neq j$, their lowest common ancestor (i.e., the deepest node that is an ancestor of both $i$ and $j$) spells a string of length $\lce(i, j)$.
\section{Lower Bounds}
\label{sec:lowerbounds}

In this section, we show lower bounds on the number of symbol comparisons required to compute a meaningful approximation of the alphabet size (\cref{sec:loweralphaapprox}) and to test square-freeness (\cref{sec:lower}).
For both bounds we use an adversarial method, which we briefly outline now.

The present model of computation may be interpreted as follows. 
An algorithm working on a string over a general unordered alphabet has no access to the actual string. 
Instead, it can only ask an oracle whether or not there are identical symbols at two positions. 
The number of questions asked is exactly the number of performed comparisons.
In order to show a lower bound on the number of comparisons required to solve some problem, we describe an adversary that takes over the role of the oracle, forcing the algorithm to perform as many symbol comparisons as possible.

\newcommand{\yesedges}{E_\textnormal{yes}}
\newcommand{\noedges}{E_\textnormal{no}}
\newcommand{\nodecol}[1]{\gamma(#1)}

We use a conflict graph $G = (V, E)$ with $V = \{1, \dots, n\}$ and $E \subseteq V^2$ to keep track of the answers given by the adversary.
The nodes directly correspond to the positions of the string. Initially, we have $E = \emptyset$ and all nodes are colorless, which formally means that they have color $\nodecol{i} = \bot$.
During the algorithm execution, the adversary may assign colors from the set $\Sigma = \{0,\dots, n - 1\}$ to the nodes, which can be seen as permanently fixing the alphabet symbol at the corresponding position (i.e., each node gets colored at most once).
The rule used for coloring nodes depends on the lower bound that we want to show (we describe this in detail in the respective sections).
Apart from this coloring rule, the general behaviour of the adversary is as follows.
Whenever the algorithm asks whether $T[i]=T[j]$ holds, the adversary answers \quotes{yes} if and only if $\nodecol{i} = \nodecol{j} \neq \bot$. 
Otherwise, it answers \quotes{no} and inserts an edge $(i, j)$ into $E$.
Whenever the adversary assigns the color of a node, it has to choose a color that is not used by any of the adjacent nodes in the conflict graph. 
This ensures that the coloring does not contradict the answers given in the past.

Let us define a set $\mathcal T \subseteq \Sigma^n$ of strings that is consistent with the answers given by the adversary.
A string $T \in \Sigma^n$ is a member of $\mathcal T$ if
%
$$\forall i \in V : \nodecol{i} \in \{\bot, T[i]\} \quad \land \quad \forall i,j \in V: (T[i] = T[j]) \implies  (i, j) \notin E.$$

Note that $\mathcal T$ changes over time. Initially (before the algorithm starts), we have $\mathcal T = \Sigma^n$. 
With every question asked, the algorithm might eliminate some strings from $\mathcal T$.
However, there is always at least on string in $\mathcal T$, which can be obtained by coloring each colorless node in a previously entirely unused color. 

%

\def\sigmaapprox{\tilde{\sigma}}

\subsection{Approximating the Alphabet Size}
\label{sec:loweralphaapprox}

Given a string $T[1..n]$ of unknown alphabet size $\sigma \geq 2$, 
assume that we want to compute an approximation of $\sigma$.
We show that if an algorithm takes at most $\frac{n\sigma} 8$ comparisons in the worst-case, then it cannot distinguish strings with at most $\sigma$ distinct symbols from strings with at least $\frac n 2$ distinct symbols.
Thus, any meaningful approximation of $\sigma$ requires $\Omega(n\sigma)$ comparisons.

For the sake of the proof, consider an algorithm that performs at most $\frac{n\sigma} 8$ comparisons when given a length-$n$ string with at most $\sigma \geq 2$ distinct symbols. We use an adversary as described at the beginning of \cref{sec:lowerbounds}, and ensure that the set $\mathcal T$ of strings consistent with the adversary's answers always contains a string with at most $\sigma$ distinct symbols. Thus, the algorithm terminates after at most $\frac{n\sigma} 8$ comparisons. At the same time, we ensure that $\mathcal T$ also contains a string with at least $\frac n 2$ distinct symbols, which yields the desired result. The adversary is equipped with the following coloring rule.
All colors are from $\{1, \dots, \sigma\}$. Whenever the degree of a node in the conflict graph becomes $\sigma - 1$, we assign its color. 
We avoid the colors of the $\sigma - 1$ adjacent nodes in the conflict graph.
At any moment in time, we could hypothetically complete the coloring by assigning one of the colors $\{1, \dots, \sigma\}$ to each colorless node, avoiding the colors of adjacent nodes.
This way, each node gets assigned one of the $\sigma$ colors, which means that $\mathcal T$ contains a string with at most $\sigma$ distinct symbols.
It follows that the algorithm terminates after at most $\frac{n\sigma} 8$ comparisons.
Each comparison may increase the degree of two nodes by one. Thus, after $\frac{n\sigma} 8$ comparisons, there are at most $\frac{n\sigma} 8 \cdot \frac 2{\sigma - 1} \leq \frac{n}{2}$ nodes with degree at least $\sigma - 1$. Therefore, at least $\frac{n}{2}$ nodes are colorless. We could hypothetically color them in $\frac{n}{2}$ distinct colors, which means that $\mathcal T$ contains a string with at least $\frac{n}{2}$ distinct symbols. This leads to the following result.

\begin{theorem}
\label{thm:inapproxalph}
For any integers $n$ and $\sigma$ with $2 \leq \sigma < \frac{n}2$, there is no deterministic algorithm that performs at most $\frac{n \sigma}8$ equality-comparisons in the worst case, and is able to distinguish length-$n$ strings with at most $\sigma$ distinct symbols from length-$n$ strings with at least $\frac{n}2$ distinct symbols. 
\end{theorem}

The theorem implies lower bounds on the number of comparisons needed to compute the LZ factorisation (as defined in \cref{sec:prelim}) and the $f$-factorisation.
In the unique $f$-factorisation $T = f_1f_2\dots f_z$, each factor $f_i$ is either a single symbol that does not occur in $f_1\dots f_{i - 1}$, or it is the fragment of maximal length such that $f_{i}$ occurs twice in $f_1\dots f_{i}$. 

\begin{corollary}
\label{cor:f-facto}
For any integers $n$ and $\sigma$ with $2 \leq \sigma < \frac{n}4$, there is no deterministic algorithm that performs at most %
$\frac{(n - 1)\sigma}{16}$ %
equality-comparisons in the worst case, and computes the $f$-factorisation of a length-$n$ string with at most $\sigma$ distinct symbols.
\end{corollary}
\begin{proof}
For some string $T = T[1]T[2]\dots T[\frac n 2]$ with $\sigma$ distinct symbols, consider the length-$n$ string $T'=T[1] T[1] T[2] T[2] \dots T[\frac n 2] T[\frac n 2]$ with $\sigma$ distinct symbols constructed by doubling each character of $T$. The alphabet size of $T$ is exactly the number of length-one phrases in the $f$-factorisation of $T'$ starting at odd positions in $T'$. Thus, by \cref{thm:inapproxalph}, we need $\frac{n \sigma}{16} = \frac{|T|\sigma}{8}$ comparisons to find the $f$-factorisation of $T'$. 
We assumed that $n$ is even, and account for odd $n$ by adjusting the bound to $\frac{(n - 1) \sigma}{16}%
$.
\end{proof}

\begin{corollary}
\label{cor:LZ}
For any integers $n$ and $\sigma$ with $3 \leq \sigma < \frac{n}6 + 1$, there is no deterministic algorithm that performs at most %
$\frac{(n - 2)(\sigma-1)}{24}$ %
equality-comparisons in the worst case, and computes the Lempel-Ziv factorisation of a length-$n$ string with at most $\sigma$ distinct symbols.
\end{corollary}
\begin{proof}
For some string $T = T[1]T[2]\dots T[\frac n 3]$ with $\sigma - 1$ distinct symbols, let $T'$ be the length-$n$ string with $\sigma$ distinct symbols constructed by doubling every character of $T$ with a separator in between, i.e., $T'=T[1]T[1] \# T[2]T[2] \# \dots \# T[\frac n 3]T[\frac n 3]\#$. The first occurrence of character $x$ in $T$ corresponds to the first occurrence of $xx\#$ in $T'$, thus the preceding phrase (possibly of length one) ends at the
first $x$ in the first occurrence of $xx\#$, and the subsequent phrase must be $x\#$. Then, for the later occurrences of $xx\#$ we cannot have that $x\#$ is a phrase.
Consequently, the alphabet size of $T$ is exactly the number of length-two phrases in the Lempel-Ziv factorisation of $T'$ starting at positions $i \equiv 2 \pmod 3$ in $T'$.
Thus, by \cref{thm:inapproxalph}, we need $\frac{n(\sigma - 1)}{24} = \frac{|T|(\sigma-1)}{8}$ comparisons to find the LZ factorisation of $T'$.
We assumed that $n$ is divisible by 3, and account for this by adjusting the bound to $\frac{(n - 2) (\sigma - 1)}{24}%
$.
\end{proof}

\subsection{Testing Square-Freeness}
\label{sec:lower}

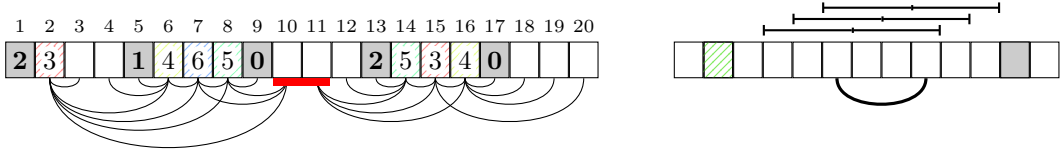
\begin{figure}
\small
\centering
\newcommand{\cgedge}[3][]{%
\path ($.5*(#2.south) + .5*(#3.south)$) ++(0, -.5) node (v) {};
\draw[#1] (#2.south) to[bend right=90] (#3.south);
}
\newlength{\nodesep}
\setlength{\nodesep}{1.8pt}
\begin{tikzpicture}
\node[inner sep=0] (prev) {};
\foreach[count=\i from 1] \x in {2,3,-1,-1,1,4,6,5,0,-1,-1,-1,2,5,3,4,0,-1,-1,-1} {
\ifnum\x<0%
\node[draw=black, right=0 of prev, inner xsep=\nodesep] (\i) {\small\phantom{I}\phantom{I}};
\else\ifnum\x<3%
\node[fill=black!20!white, draw=black, right=0 of prev, inner xsep=\nodesep] (\i) {\small\phantom{I}\clap{\boldmath{$\x$}}\phantom{I}};
\else\ifnum\x>2%
\pgfmathtruncatemacro{\col}{(\x-3)*48}
\definecolor{cccol}{HSB}{\col,100,240}
\node[draw=black, pattern=north east lines, pattern color=cccol, right=0 of prev, inner xsep=\nodesep] (\i) {\small\phantom{I}%
\clap{\contourlength{2pt}\contour{white}{$\x$}}%
\phantom{I}};
\fi\fi\fi
\node[above=0 of \i] {$\scriptstyle\i$};
\node[left=0 of \i.east] (prev) {};
}


\cgedge{2}{3}
\cgedge{2}{6}
\cgedge{2}{7}
\cgedge{2}{8}
\cgedge{2}{10}
\cgedge{4}{6}
\cgedge{5}{6}
\cgedge{6}{7}
\cgedge{7}{8}
\cgedge{7}{10}
\cgedge{8}{9}
\cgedge{8}{10}

\cgedge{11}{14}
\cgedge{11}{15}
\cgedge{11}{16}
\cgedge{12}{14}
\cgedge{13}{14}
\cgedge{14}{15}

\cgedge{15}{16}
\cgedge{15}{20}
\cgedge{16}{17}
\cgedge{16}{18}
\cgedge{16}{19}

\path (11.south east) ++(-0.02, -.1) node (tr) {};
\fill[red] (10.south west) ++(0.02, 0) rectangle (tr.center);

\path (prev) ++(1,0) node (prev) {};
\foreach \i in {1,...,13} {
\node[draw=black, right=0 of prev, inner xsep=\nodesep] (\i) {\small\phantom{I}\phantom{I}};
\node[above=0 of \i] {\color{white}$\scriptstyle\i$};
\node[left=0 of \i.east] (prev) {};
}

\node[fill=black!20!white, draw=black, inner xsep=\nodesep] (12) at (12) {\small\phantom{I}\phantom{I}};
\definecolor{cccol}{HSB}{70,140,220}
\node[draw=black, pattern=north east lines, pattern color=cccol, inner xsep=\nodesep] (2) at (2) {\small\phantom{I}\phantom{I}};

\cgedge[very thick]{6}{9}

\foreach[count=\i from 1] \x in {4,5,6} {
\pgfmathtruncatemacro{\y}{\x + 3}
\pgfmathtruncatemacro{\z}{\y + 3}
\pgfmathsetmacro{\v}{\i*.15}

\path (\x.north) ++(0, \v) node (v) {};
\draw[{|[width=4pt]}-, thick] (\x.west |- v) to (\y.west |- v);
\draw[{|[width=2pt]}-{|[width=4pt]}, thick] (\y.west |- v) to (\z.west |- v);

}

\end{tikzpicture}
\caption{Example conflict graph of the adversary described in \cref{sec:lower}. 
The alphabet $\{0, \dots, 15\}$ is of size $\sigma = 16$. 
The blocks are of length $\frac \sigma 4 = 4$. 
The gray nodes are exactly the starting positions of the blocks and contain the symbols of the ternary Thue-Morse sequence $v = 2,1,0,2,0,1,2,\dots$, which is square-free. 
We assume that the colored nodes were colored in the following order: $2,6,8,7,15,16,14$. 
At the time of coloring node $8$, we had to avoid colors 
$0,1,2$ (because they are reserved for the separator positions), 
$3$ (because the adjacent node $2$ already has color $3$), and 
$4$ (because node $6$ is in the same block and already has color $4$). 
The algorithm has not eliminated all squares yet. 
For example, nodes 10 and 11 with absent edge $(10, 11) \notin E$ are adjacent to nodes of colors $\{3,6,5\} \cup \{5,3,4\}$. 
Thus, any of the colors $\{0,1,2\} \cup \{7, \dots, 15\}$ can be assigned to both nodes, enforcing the square $T[10..11]$. 
As visualized on the right, an edge of length $\ell$ eliminates at most $\ell$ squares.} 
\label{fig:lbsquarefree}
\end{figure}

In this section, we prove that testing square-freeness requires at least $n \ln \sigma - 3.6n$ comparisons (even if $\sigma$ is known). 
The proof combines the idea behind the original $\Omega(n \lg n)$ lower bound by Main and Lorentz \cite{Main1984} with the adversary described at the beginning of \cref{sec:lowerbounds}.
This time, we ensure that $\mathcal T$ always contains a square-free string with at most $\sigma$ distinct symbols.
At the same time, we try to ensure that $\mathcal T$ also contains a string with at least one square.
We will show that we can maintain this state until at least $n \ln \sigma -3n$ comparisons have been performed.

The string (or rather family of strings) constructed by the adversary is organized in $\ceil{\frac {4n} \sigma}$ non-overlapping blocks of length $\frac \sigma 4$ (we assume $\frac \sigma 4 \in \mathbb N$ and $8 \leq \sigma \leq n$).
Each block begins with a special separator symbol. 
More precisely, the first symbol of the $k$-th block is the $k$-th symbol of a ternary square-free word over the alphabet $\{0,1,2\}$ (e.g., the distance between the $k^{\text{th}}$ and $(k+1)^{\text{th}}$ occurrence of 0 in the Prouhet-Thue-Morse sequence, also known as the ternary Thue-Morse-Sequence, see~\cite[Corollary~1]{Allouche1998}).
Initially, the adversary colors the nodes that correspond to the separator positions in their respective colors from $\{0,1,2\}$.
All remaining nodes will later get colors other than $\{0,1,2\}$. Any fragment crossing a block boundary can be projected on the colors $\{0,1,2\}$, and by construction the string cannot contain a square.
Thus, the separator symbols ensure that there is no square crossed by a block boundary, which implies that the string is square-free if and only if each of its blocks is square-free. 

During the algorithm execution, we use the following coloring rule. The available colors are $\{3, \dots, \sigma - 1\}$. Whenever the degree of a node becomes $\frac \sigma 4$, we assign its color. We avoid not only the at most $\frac \sigma 4$ colors of already colored neighbors in the conflict graph, but also the less than $\frac \sigma 4$ colors of nodes within the same block (due to $\sigma \geq 8$, there are at least $\sigma - 3 - \frac \sigma 2 \geq 1$ colors available). An example of the conflict graph is provided in \cref{fig:lbsquarefree}. At any moment in time, we could hypothetically complete the coloring by assigning one of the colors $\{3, \dots, \sigma - 1\}$ to each colorless node, avoiding colors of adjacent nodes and colors of nodes in the same block. Afterwards, each node holds one of the $\sigma$ colors, but no two nodes within the same block have the same color. Thus, each block is square-free, and therefore $\mathcal T$ always contains a square-free string with at most $\sigma$ distinct symbols. 

Now we consider the state of the conflict graph \emph{after the algorithm has terminated}. 
We are particularly concerned with consecutive ranges of colorless nodes.
The following lemma states that for each such range, the algorithm either performed many comparisons, or we can enforce a square within the range.

\begin{lemma}\label{lem:MLrestated}
Let $R = \{i,\dots,j\} \subset V$ be a consecutive range of $m = j - i + 1$ colorless nodes in the conflict graph. Then either $\absolute{E \cap R^2} \geq \sum_{\ell = 1}^{\floor{m / 2}} \frac{m - 2\ell + 1}{\ell}$, or there is a string $T \in \mathcal T$ with at most $\sigma$ distinct symbols such that $T[i..j]$ contains a square.
\end{lemma}
\begin{proof}
We say that an integer interval $[x,x+2\ell-1]$ with $i \leq x < (x + 2\ell - 1) \leq j$ has been \emph{eliminated}, if for some $y$ with $x \leq y < x + \ell$ there is an edge $(y, y + \ell)$ in the conflict graph.
If such an edge exists, then (by the definition of $\mathcal T$) all strings $T \in \mathcal T$ satisfy $T[y] \neq T[y + \ell]$. Thus $T[x..x+2\ell-1]$ is not a square for any of them.

Now we show that if $[x, x+2\ell-1]$ has not been eliminated, then there exists a string $T \in \mathcal T$ such that $T[x..x+2\ell-1]$ is a square.
For this purpose, consider any position $y$ with $x \leq y < x + \ell$, i.e., a position in the first half of the potential square.
Since $[x, x+2\ell-1]$ has not been eliminated, $(y, y + \ell)$ is not an edge in the conflict graph. It follows that we could assign the same color to $y$ and $y + \ell$. 
We only have to avoid the at most $2 \cdot (\frac \sigma 4 - 1)$ colors of adjacent nodes of both $y$ and $y + \ell$ in the conflict graph. 
Thus there are $\frac \sigma 2 + 2$ appropriate colors that can be assigned to both nodes. Unlike during the algorithm execution, we do not need to avoid the special separator colors or the colors in the same block; since we are trying to enforce a square, we do not have to worry about accidentally creating one.
By applying this coloring scheme for all possible choices of $y$, we enforce that all strings $T \in \mathcal T$ have a square $T[x..x+2\ell-1]$. 
Note that by coloring additional nodes after the algorithm terminated, we only remove elements from $\mathcal T$. Thus, the strings with square $T[x..x+2\ell-1]$ were already in $\mathcal T$ when the algorithm terminated.
It follows that, if the algorithm actually guarantees square-freeness, then it must have eliminated all possible intervals $[x,x+2\ell-1]$ with $i \leq x < (x + 2\ell - 1) \leq j$.

While each interval needs at least one edge to be eliminated, a single edge eliminates multiple intervals. 
However, all the intervals eliminated by an edge must be of the same length.
Now we give a lower bound on the number of edges needed to eliminate all intervals of length $2\ell$.
Any edge $(y, y + \ell)$ eliminates $\ell$ intervals, namely the intervals $[x,x+2\ell-1]$ that satisfy $x \leq y < x + \ell$. Within $R$, we have to eliminate $m - 2\ell + 1$ intervals of length $2\ell$, namely the intervals $[x,x+2\ell-1]$ that satisfy $i \leq x \leq j - 2\ell + 1$ (see right side of \cref{fig:lbsquarefree}). Thus we need at least $\frac{m - 2\ell + 1}{\ell}$ edges to eliminate all squares of length $2\ell$.
Finally, by summing over all possible values of $\ell$, we need at least $\sum_{\ell = 1}^{\floor{m / 2}} \frac{m - 2\ell + 1}{\ell}$ edges to eliminate all intervals in $R$.
Note that the edges used for elimination have both endpoints in $R$, and are thus contained in $E \cap R^2$.
Consequently, if $\absolute{E \cap R^2} < \sum_{\ell = 1}^{\floor{m / 2}} \frac{m - 2\ell + 1}{\ell}$, then not all intervals have been eliminated, and there is a string in $\mathcal T$ that contains a square.
\end{proof}

Finally, we show that the algorithm either performed at least $\Omega(n \lg \sigma)$ comparisons, or there is a string $T \in \mathcal T$ that contains a square.
Let $c_1, c_2, \dots, c_k$ be exactly the colored nodes. Initially (before the algorithm execution), the adversary colored $\ceil{\frac {4n} \sigma}$ nodes. Thus $k \geq \ceil{\frac {4n} \sigma}$, and there are $k - \ceil{\frac {4n} \sigma}$ nodes that have been colored after their degree reached $\frac \sigma 4$. Therefore, the sum of degrees of all colored nodes is at least $(k - \ceil{\frac {4n} \sigma}) \cdot \frac \sigma 4 \geq \frac{\sigma k - 4n - \sigma}{4} \geq \frac{\sigma k - 5n}{4}$. Each comparison may increase the degree of two nodes by one. Thus, the colored nodes account for at least $\frac{\sigma k - 5n}{8}$ comparisons.
There are $k$ non-overlapping maximal colorless ranges of nodes, namely $\{c_i + 1, \dots, c_{i + 1} - 1\}$ for $1 \leq i \leq k$ with auxiliary value $c_{k + 1} = n + 1$. According to \cref{lem:MLrestated}, each respective range accounts for $e_i = \sum_{\ell = 1}^{\floor{m_i / 2}} \frac{m_i - 2\ell + 1}{\ell}$ edges, where $m_i = c_{i + 1} - c_i - 1$. (No edge gets counted more than once because the ranges are non-overlapping, and both endpoints of the respective edges are within the range.)
Thus, in order to verify square-freeness, the algorithm must have performed at least $\sum_{i = 1}^k e_i + \frac{\sigma k - 5n}{8}$ comparisons. 
The remainder of the proof consists of simple algebra. First, we provide a convenient lower bound for $e_i$ (explained below):%

\begin{alignat*}{1}\label{eqn:MLlower}
e_i = \sum_{\ell = 1}^{\floor{m_i / 2}} \frac{m_i - 2\ell + 1} \ell = \sum_{\ell = 1}^{\ceil{m_i / 2}} \frac{m_i - 2\ell + 1} \ell \geq\enskip &%
(m_i + 1) \left(\left(\sum_{\ell = 1}^{\ceil{m_i / 2}} \frac 1 \ell \right) - 1\right)\\
>\enskip &%
(m_i + 1) \cdot(\ln \frac {m_i} 2 - \frac 1 2)\\
=\enskip &%
(m_i + 1) \cdot\ln \frac {m_i}{2\sqrt{e}}\\
\geq\enskip &%
(m_i + 1) \cdot\ln \frac {m_i + 1}{2.5\sqrt{e}}%
\end{alignat*}

We can replace $\floor{m_i / 2}$ with $\ceil{m_i / 2}$ because if $m_i$ is odd the additional summand equals zero.
The first inequality uses simple arithmetic operations. 
The second inequality uses the classical lower bound $(\ln x + \frac 1 2) < H_x$ of harmonic numbers. 
The last inequality holds for $m_i \geq 4$. 
For $m_i < 4$ the result becomes negative and is thus still a correct lower bound for the number of comparisons.
We obtain:

\def\numsubstr{k}%
\begin{alignat*}{1} 
\smash{%
\underbrace{\sum_{i=1}^\numsubstr (m_i + 1) \cdot \ln\frac {m_i + 1}{2.5\sqrt{e}}}%
_{\text{\begin{tabular}{c}comparisons within colorless ranges\end{tabular}}} + %
\underbrace{\vphantom{\sum_{i=1}^\numsubstr}\frac{\sigma k - 5n} 8}%
_{\text{\begin{tabular}{c}\vphantom{lbthi}%
comparisons for\\[-.5em]%
colored nodes\end{tabular}}}%
}\quad \geq\quad & n \cdot \ln \frac{n}{2.5\sqrt{e}\numsubstr} + \frac{\sigma k - 5n} 8\\
=\quad &n \cdot \ln \frac \sigma {2.5\sqrt{e}x} + \frac{xn - 5n} 8\\
=\quad &n \cdot \ln \sigma + n \cdot \left(\frac{x - 5} 8 - \ln 2.5\sqrt{e}x\right)\\
>\quad &n \cdot \ln \sigma - 3.12074n
\end{alignat*}%

The first step follows from $\sum_{i=1}^\numsubstr (m_i + 1) = n$ and the log sum inequality (see \cite[Theorem 2.7.1]{Cover2006}). In the second step we replace $\numsubstr$ by using $x = \frac{\sigma \numsubstr} n$. The third step uses simple arithmetic operations. The last step is reached by substituting $x = 8$, which minimizes the equation.
Finally, we assumed that $\sigma$ is divisible by 4. We account for this by adjusting the lower bound to $n \ln (\sigma - 3) - 3.12074n$, which is larger than $n \ln \sigma - 3.6n$ for $\sigma \geq 8$.


\lowerbound*

\label{sec:upper}

In this section, we consider the problem of testing square-freeness of a given string.
We introduce an algorithm that decides whether or not a string is square-free using only $\Oh(n \lg \sigma)$ comparisons, matching the lower bound from \cref{sec:lower}.
Note that this algorithm is not yet time efficient because, apart from the performed symbol comparisons, it uses other operations that are expensive in the Word RAM model.
A time efficient implementation of the algorithm will be presented in \cref{sec:alg}, where we first achieve $\Oh(n \lg \sigma + n \log^* n)$ time, and then improve this to $\Oh(n \lg \sigma)$ time.
In \cref{sec:runs}, we generalize the result to compute all runs in the same time complexity.

\subsection{Sparse Suffix Trees and Difference Covers}

\begin{lemma}
\label{lem:sparse}
The sparse suffix tree containing any $b$ suffixes $T[i_{1}..n], \ldots, T[i_{b}..n]$
of $T[1..n]$ can be constructed using $\Oh(b\sigma\log b)$ comparisons plus
$\Oh(n)$ comparisons shared by all invocations of the lemma.
\end{lemma}

\begin{proof}
We maintain a union-find structure over the positions of $T[1..n]$. Initially, each position is in a separate
component. Before issuing a query $T[x]\stackrel{?}{=}T[y]$, we check if $x$ and $y$ are in the same
component of the union-find structure, and if so immediately return that $T[x]=T[y]$
without performing any comparisons. 
Otherwise,
we issue the query and if it returns that $T[x]=T[y]$ we merge the components of $x$ and $y$.
Thus, the total number of issued queries with positive answer, over all invocations of the lemma,
is less than $n$, and it remains to bound the number of issued queries with negative answer.

We insert the suffixes $T[i_{j}..n]$ one-by-one into an initially empty sparse suffix tree.
To insert the next suffix, we descend from the root of the tree to identify the node $u$ that corresponds to
the longest common prefix between $T[i_{j}..n]$ and any of the already inserted suffixes. We then make $u$ explicit unless it is explicit already,
and add an edge from $u$ to a new leaf corresponding to the whole $T[i_{j}..n]$.
We say that the insertion procedure terminates at $u$.
Node $u$ can be identified with only $\Oh(\sigma\log b)$ comparisons with negative answers as follows.
Let $v$ be the current node (initially, the root of the tree), and let $v_{1},\ldots,v_{d}$ be its children,
where $d\leq \sigma$. Here, $v$ can be either explicit or implicit, in the latter case $d=1$.
We arrange the children of $v$ so that the number of leaves in the subtree
rooted at $v_{1}$ is at least as large as the number of leaves in the subtree rooted at any other
child of $v$. Then, we compare the character on the edge leading to $v_{1}$ with
the corresponding character of the current suffix. If they are equal we continue with $v_{1}$,
otherwise we compare the characters on the edges leading to $v_{2},\ldots,v_{d}$ with the
corresponding character of the current suffix one-by-one. Then, we either continue
with some $v_{j}$, $j\geq 2$, or terminate at $v$. To bound the number of comparisons with negative
answer, observe that such comparisons only occur when we either terminate at $v$ or continue
with $v_{j}$, $j\geq 2$. Whenever we continue with $v_{j}$, $j\geq 2$, the number of leaves
in the current subtree rooted at $v_j$ decreases at least by a factor of 2 compared to subtree rooted at $v$ (as the subtree rooted at $v_{1}$ had the largest
number of leaves). Thus, during the whole descent from the root performed during an insertion
this can happen only at most $1+\log b$ times. Every time we do not continue in the subtree $v_{1}$
we might have up to $d\leq \sigma$ comparisons with negative answer, thus the total number of such comparisons is as claimed%
\footnote{In the descent, if all children are sorted according to their subtree size, the number of comparisons decreases to $\Oh(b(\sigma/\log \sigma)\log b)$, but this appears irrelevant for our final algorithm.}.
\end{proof}

Now we describe the sample positions that we will later use to compute the approximate LZ factorisation. 
A set ${\bf S}\subseteq \mathbb N$ is called a $t$-cover of $\{1,\ldots,n\}$ if there is a constant-time computable function $h$
such that, for any $1\leq i,j\leq n-t+1$, we have $0\leq h(i,j)<t$ and $i+h(i,j),j+h(i,j)\in {\bf S}$.
A possible construction of $t$-covers is given by the lemma below, and visualized in \cref{fig:diff-cover}.

\begin{lemma}
\label{lem:cover}
For any $n$ and $t\leq n$, there exists a $t$-cover $\D(t)$ of $\{1,\ldots,n\}$ with size $\Oh(n/\sqrt{t})$.
Furthermore, its elements can be enumerated in time proportional to their number.
\end{lemma}

\begin{proof}
We use the well-known combinatorial construction known as difference covers, see e.g.~\cite{Maekawa1985}.
Let $r=\lfloor \sqrt{t}\rfloor$ and define $\D(t)=\{ i\in\{1,\ldots,n\} : i\bmod r = 0 \text { or } i\bmod r^{2} \in \{ 0,\ldots,r-1 \} \}$.
By definition, $|\D(t)| \leq \lfloor n/r \rfloor + \lfloor n/r^{2} \rfloor r = \Oh(n/r) = \Oh(n/\sqrt{t})$.
The function $h(i,j)$ is defined as $a+b\cdot r$, where $a=(r-i)\bmod r$ and $b=(r-\lfloor (j+a)/r \rfloor )\bmod r$.
Note that $i+h(i,j) \leq n$ and $j+h(i,j) \leq n$.
Then, $i+(a+b\cdot r) = 0 \pmod r$, while $\lfloor (j+(a+b\cdot r))/r \rfloor = \lfloor (j+a)/r + b\rfloor = 0 \bmod r$
implies $j+h(i,j) \bmod r^{2} \in \{ 0,\ldots,r-1 \} \}$,
thus $i+h(i,j),j+h(i,j)\in \D(t)$ as required.
\end{proof}

\newcommand{\deltacover}[2]{
\node (p) at #1 {};
\draw (p) rectangle ++(#2,-0.5); 
\foreach \i in {1,...,#2} {
	\draw ($(p)+(\i,0)$) -- ($(p)+(\i,-0.5)$);
}
\foreach \i in {1,...,8} {
	\draw ($(p)+(\i*0.125,0)$) -- ($(p)+(\i*0.125,-0.5)$);
}
}

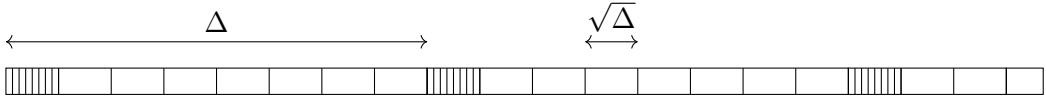
\begin{figure}
\centering
\begin{tikzpicture}[scale=0.7, every node/.style={scale=1}]
\deltacover{(0,0)}{8}
\draw[<->] (0,0.5) -- (8,0.5) node[midway,above] (delta) {$\Delta$};
\draw[<->] (11,0.5) -- (12,0.5) node[midway,above] (sqrt_delta) {$\sqrt{\Delta}$};
\deltacover{(8,0)}{8}
\deltacover{(16,0)}{3}
\draw (19,0) rectangle ++(0.7,-0.5); 
\end{tikzpicture}
\caption{Positions in a $\Delta$-difference cover.}
\label{fig:diff-cover}
\end{figure}

\subsection[Detecting Squares with a Delta-Approximate LZ Factorisation]{Detecting Squares with a \boldmath$\Delta$\unboldmath-Approximate LZ Factorisation}

A crucial notion in our algorithm is the following variation on the standard Lempel-Ziv factorisation:

\begin{definition}[$\Delta$-approximate LZ factorisation]
For a positive integer parameter $\Delta$, the fragment $T[s..e]$ is a $\Delta$-approximate
LZ phrase if it can be split into a head and a tail
$T[s..e]=\head(T[s..e])\tail(T[s..e])$ such that $\absolute{\head(T[s..e])} < \Delta$ and additionally
\begin{itemize}
\item $\tail(T[s..e])$ is either empty or occurs at least twice in $T[1..e]$, and
\item the unique (standard) LZ phrase $T[s..e']$ starting at position $s$ satisfies $e'-1\leq e$.
\end{itemize}
In a $\Delta$-approximate LZ factorisation $T = b_1b_2\dots b_z$, each factor $b_i$ is a $\Delta$-approximate phrase $T[s..e]$ with $s = 1+\sum_{j = 1}^{i - 1} \absolute{b_j}$ and $e = \sum_{j = 1}^{i} \absolute{b_j}$.
\end{definition}
Note that a standard LZ phrase is not a $\Delta$-approximate phrase. Also, while the LZ phrase starting at each position (and thus also the LZ factorisation) is uniquely defined, there may be multiple different $\Delta$-approximate phrases starting at each position. This also means that a single string can have multiple different $\Delta$-approximate factorisations. 
The definitions of both standard and $\Delta$-approximate LZ phrases are illustrated in \cref{fig:def_phrases}.
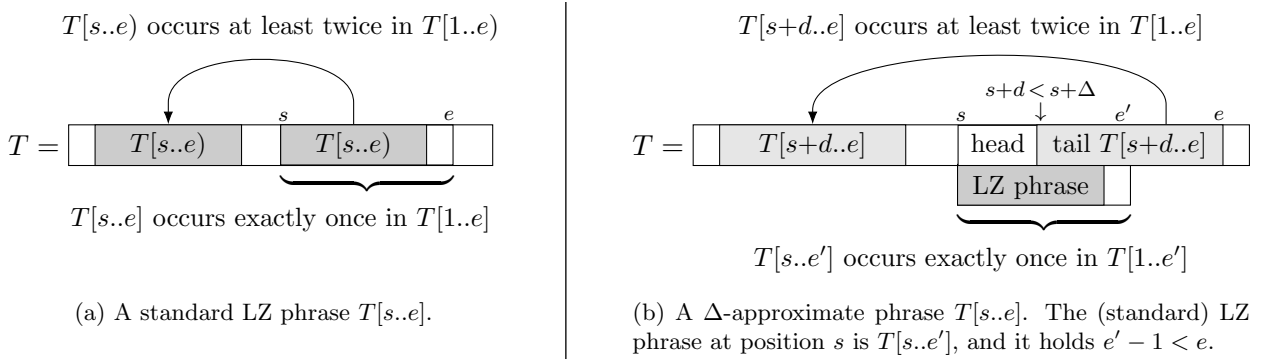
\begin{figure}
\subcaptionbox{\label{fig:LZ-phrase}A standard LZ phrase $T[s..e]$.}%
{\begin{tikzpicture}[x=10, y=15, every node/.style={inner sep=0, minimum width=0, minimum height=0}]
\node (l) at (2,0) {};
\node (r) at (18,1) {};
\node (s) at (10,0) {};
\node (e) at (16.5,1) {};
\node (shift) at (7,0) {};

\node (eminus) at ($(e)-(1,0)$) {};
\node (mid) at ($.5*(l)+.5*(r)$) {};

\node[draw, fit=(l)(r)] (text) {};
\node[left=0 of text] {$T = {}$};
\node[above right=.1 and -.1em of s |- e] (toplabel) {$\scriptstyle s$};
\node[above left=.1 and -.1em of e] {$\scriptstyle e$};

\node[draw, fit=(s |- l)(e)] (phrase) {};
\node[draw, fill=gray!40, fit=(s |- l)(eminus)] (phraseminus) {};
\node at (phraseminus.center) {\small $T[s..e)$};

\node[below=.2em of phrase.center |- text.south] (bottomtext) {$\underbrace{\hspace{2.3cm}}$};

\draw ($(s)-(shift)$) node (s2) {};
\draw ($(eminus)-(shift)$) node (eminus2) {};
\node[draw, fill=gray!40, fit=(s2)(eminus2)] (phraseminus2) {};
\node at (phraseminus2.center) {\small $T[s..e)$};

\draw[-Latex] (phraseminus.north) to ++(0,0.75em) to[out=90, in=90, looseness=.8] node[midway] (toptext) {} ($(phraseminus2.north)+(0,0.75em)$) to (phraseminus2.north);

\path (text) ++(0,2.25) node (ttv) {};
\node (toptext) at (toptext |- ttv) {};

\node[below=.5em of mid |- bottomtext.south] {\small $T[s..e]$ occurs exactly once in $T[1..e]$};
\node[above=.5em of mid |- toptext.north] {\small $T[s..e)$ occurs at least twice in $T[1..e)$};

\node at (0,-3) {};
\end{tikzpicture}}%
~\hfill\rule[-1cm]{.5pt}{4.75cm}\hfill~%
\subcaptionbox{\label{fig:Delta-approx}A $\Delta$-approximate phrase $T[s..e]$. The (standard) LZ phrase at position $s$ is $T[s..e']$, and it holds $e' - 1 < e$.}%
{\begin{tikzpicture}[x=10, y=15, every node/.style={inner sep=0, minimum width=0, minimum height=0}]
\node (l) at (-0,0) {};
\node (r) at (21,1) {};
\node (s) at (10,0) {};
\node (sd) at (13,0) {};
\node (ep) at (16.5,-1) {};
\node (e) at (20,1) {};
\node (shift) at (12,0) {};

\node (epminus) at ($(ep)-(1,0)$) {};
\node (mid) at ($.5*(l)+.5*(r)$) {};
\node (lzmid) at ($.5*(s)+.5*(ep)$) {};

\node[draw, fit=(l)(r)] (text) {};
\node[left=0 of text] {$T = {}$};
\node[above right=.1 and -.1em of s |- e] (toplabel) {$\scriptstyle s$};
\node[above right=.1 and -.1em of sd |- e] {$\overset{\mathclap{s{+}d\,<\,s{+}\Delta}}{\scriptstyle\downarrow}$};
\node[above left=.1 and -.1em of e] {$\scriptstyle e$};
\node[above left=.1 and -.1em of ep |- e] {$\scriptstyle e'$};

\node[draw, fit=(s)(ep)] (phrase) {};
\node[draw, fill=gray!40, fit=(s)(epminus)] (phraseminus) {};
\node at (phraseminus.center) {\small LZ \smash{phrase}};

\node[draw, fill=white, fit=(s |- r)(sd)] (head) {};
\node[draw, fill=gray!20, fit=(sd)(e)] (tail) {};

\node at (head.center) {\small head\vphantom{$[]$}};
\node at (tail.center) {\small tail $T[s{+}d..e]$};

\node[below=.2em of lzmid |- phrase.south] (bottomtext) {$\underbrace{\hspace{2.3cm}}$};

\draw ($(sd)-(shift)$) node (sd2) {};
\draw ($(e)-(shift)$) node (e2) {};
\node[draw, fill=gray!20, fit=(sd2)(e2)] (tail2) {};
\node at (tail2.center) {\small $T[s{+}d..e]$};

\draw[-Latex] (tail.30) to ++(0,0.5) to[out=90, in=90, looseness=.5] node[midway] (toptext) {} ($(tail2.north)+(0,0.5em)$) to (tail2.north);

\path (text) ++(0,2.25) node (ttv) {};
\node (toptext) at (toptext |- ttv) {};

\node[below=.5em of mid |- bottomtext.south] {\small $T[s..e']$ occurs exactly once in $T[1..e']$};
\node[above=.5em of mid |- toptext.north] {\small $T[s{+}d..e]$ occurs at least twice in $T[1..e]$};

\node at (0,-3) {};
\end{tikzpicture}}

\caption{Illustration of the definition of a LZ-phrase and a $\Delta$-approximate phrase.}\label{fig:def_phrases}
\end{figure}

The intuition behind the above definition is that constructing the $\Delta$-approximate LZ factorisation becomes easier for larger
values of $\Delta$. In particular, for $\Delta=n$ one phrase is enough. We formalise this in the following lemma,
which is made more general for the purpose of obtaining the final result in this section.

\begin{lemma}
\label{lem:compute}
For any parameter $\Delta \in [1, m]$, a $\Delta$-approximate LZ factorisation of any fragment $T[x..y]$ of length $m$ can be
computed with $\Oh(m\sigma\log m/\sqrt{\Delta})$ comparisons plus $\Oh(n)$ comparisons shared by all invocations of the lemma.
\end{lemma}

\begin{proof}
By \cref{lem:cover}, there exists a $\Delta$-cover $\D(\Delta)$ of $\{1,\ldots,n\}$ with size $\Oh(n/\sqrt{\Delta})$.
Let $S=\D(\Delta)\cap\{x,x+1,\ldots,y\}$. Let $S=\{i_{1},i_{2},\ldots,i_{b}\}$.
It is straightforward to verify that the construction additionally guarantees
$b = \Oh(m/\sqrt{\Delta})$. We apply \cref{lem:sparse} on the suffixes $T[i_{1}..n],\ldots,T[i_{b}..n]$ to obtain
their sparse suffix tree $T$ with $\Oh(b\sigma\log b)$ comparisons plus $\Oh(n)$ comparisons shared by all invocations of the lemma.
$T$ allows us to obtain the longest common prefix of any two fragments $T[i..y]$ and $T[j..y]$, for
$i,j\in S$, with no additional comparisons.
By the properties of $\D(\Delta)$, for any $i,j\in \{x,x+1,\ldots,y-\Delta+1\}$ we have $0\leq h(i,j)<\Delta$
and $i+h(i,j), j+h(i,j)\in S$.

We compute the $\Delta$-approximate LZ factorisation of $T[x..y]$ phrase-by-phrase. Denoting the remaining suffix
of the whole $T[x..y]$ by $T[x'..y]$, we need to find $x' \leq y'\leq y$ such that $T[x'..y']$ is a 
$\Delta$-approximate phrase. This is done as follows.
We iterate over every $x'\leq x'' < x'+\Delta$ such that $x''\in S$.
For every such $x''$, we consider every $x\leq a' < x'$ such that $a'\in S$,
and compute the length $\ell$ of the longest common prefix of $T[x''..y]$ and $T[a'..y]$.
Among all such $x'',a'$ we choose the pair that results in the largest value of $x''-x'+\ell-1$
and choose the next phrase to be $T[x'..(x''+\ell-1)]$, with the head being $T[x'..(x''-1)]$ and the tail $T[x''..(x''+\ell)-1]$.
Finally, if there is no such pair, or the value of $x''-x'+\ell-1$ corresponding to the found pair is less than $\Delta-2$,
we take the next phrase to be $T[x'..\min\{x'+\Delta-1,y\}]$ (with empty tail).
Selecting such a pair requires no extra comparisons, as for every $x'',a'\in S$ we can
use the sparse suffix tree to compute $\ell$.
While it is clear that the generated $\Delta$-approximate phrase has the required form,
we need to establish that it is sufficiently long. 

Let $T[x'..y'']$ be the (unique) standard LZ phrase of $T[x..y]$ that is prefix of $T[x'..y]$. 
If $y''<x'+\Delta-1$ then we only need to ensure that the generated $\Delta$-approximate phrase
is of length at least $\min\{\Delta-1, y-x'+1\}$, which is indeed the case.
Therefore, it remains to consider the situation when $y''\geq x'+\Delta-1$.
Let $T[a..b]$ be the previous occurrence of $T[x'..(y''-1)]$ in $T[x..y]$
(because $T[x'..y'']$ is a phrase this is well defined).
Thus, $T[a..b]=T[x'..(y''-1)]$ and $a<x'$.
Because $y'' \geq x'+\Delta-1$ and $y''\leq y$, as explained above $0\leq h(a,x')<\Delta$
and $a+h(a,x'), x'+h(a,x')\in S$.
We will consider $x''=x'+h(a,x')$ and $a'=a+h(a,x')$ in the above procedure.
Next, $T[a'..b]=T[x''..(y''-1)]$,
so when considering this pair we will obtain $\ell \geq |T[x''..(y''-1)]|$.
Thus, for the found pair we will have $x''+\ell-1 \geq y''-1$ as required
in the definition of a $\Delta$-approximate phrase.
\end{proof}

Next, we show that even though the $\Delta$-approximate LZ factorisation does not capture all distinct squares, as it is the case for the standard LZ factorisation, it is still helpful in detecting all sufficiently long squares.
A crucial component is the following property of the $\Delta$-approximate LZ factorisation.

\begin{lemma}
\label{lem:longhelper}
Let $b_{1}b_{2}\ldots b_{z}$ be a $\Delta$-approximate LZ factorisation of a string $T$. For every square $T[s..s+2\ell - 1]$ of length $2\ell \geq 8\Delta$, there is at least one phrase $b_i$ with $\absolute{\tail(b_i)} \geq \frac \ell 4 \geq \Delta$ such that $\tail(b_i)$ and the right-hand side $T[s + \ell..s+2\ell - 1]$ of the square intersect.
\end{lemma}
\begin{proof}
Assume that all tails that intersect $T[s + \ell..s+2\ell - 1]$ are of length less than $\frac \ell 4$, then the respective phrases of these tails are of length at most $\frac \ell 4 + \Delta - 1$ (because each head is of length less than $\Delta$). 
This means that $T[s + \ell..s+2\ell - 1]$ intersects at least $\ceil{\ell / (\frac \ell 4 + \Delta - 1)} \geq \ceil{\ell / (\frac \ell 2 - 1)} = 3$ phrases.
Thus there is some phrase $b_i = T[x..y]$ properly contained in $T[s+\ell..s+2\ell-1]$, formally $s+\ell < x \leq y < s + 2\ell - 1$.
However, this contradicts the definition of the $\Delta$-approximate LZ factorisation because $T[x..s+2\ell]$ is the prefix of a standard LZ phrase (due to $T[x..s+2\ell - 1] = T[x - \ell..s+\ell - 1]$), and the $\Delta$-approximate phrase $b_i = T[x..y]$ must satisfy $y \geq s + 2\ell - 1$.
The contradiction implies that $T[s+\ell..s+2\ell - 1]$ intersects a tail of length at least $\frac \ell 4 \geq \Delta$.
\end{proof}

\begin{lemma}
\label{lem:long}
Given a $\Delta$-approximate LZ factorisation $T = b_{1}b_{2}\ldots b_{z}$,
we can detect a square of size $\geq8\Delta$ in $\Oh\left(\sum_{\absolute{\tail\left(b_{i}\right)} \geq \Delta}\absolute{\tail\left(b_{i}\right)} + z\right)$ time and $\Oh\left(\sum_{\absolute{\tail\left(b_{i}\right)} \geq \Delta}\absolute{\tail\left(b_{i}\right)}\right)$ comparisons.
\end{lemma}

\begin{proof}
\newcommand{\targettime}{\Sigma_{\textsf{target}}}
We consider each phrase $b_i = T[a_1..a_3]$ with $\head(b_i) = T[a_1..a_2 - 1]$ and $\tail(b_i) = T[a_2..a_3]$ separately. Let $k = \absolute{\tail(b_i)}$. 
If $k \geq \Delta$, we apply \cref{lem:conquer} to $x_1 = T[a_2 - 8k..a_2 - 1]$ and $y_1 = T[a_2..a_3 + 4k - 1]$, as well as $x_2 = T[a_2 - 8k..a_3 - 1]$ and $y_2 = T[a_3..a_3 + 4k - 1]$ trimmed to $T[1..n]$. This takes $\Oh(\absolute{\tail(b_i)})$ time and comparisons, or $\Oh\left(\sum_{\absolute{\tail\left(b_{i}\right)} \geq \Delta}\absolute{\tail\left(b_{i}\right)}\right)$ time and comparisons for all phrases. We need additional $\Oh(z)$ time to check if $k \geq \Delta$ for each phrase.

Now we show that the described strategy detects a square of size at least $8\Delta$. Let $T[s..s+2\ell - 1]$ be any such square.
Due to \cref{lem:longhelper}, the right-hand side $T[s+\ell..s+2\ell - 1]$ of this square intersects some tail $\tail(b_i) = T[a_2..a_3]$ of length $k = \absolute{\tail(b_i)} \geq \frac \ell 4 \geq \Delta$.
Due to the intersection, we have $a_2 \leq s +2\ell - 1$ and $a_3 \geq s + \ell$.
Thus, when processing $b_i$ and applying \cref{lem:conquer}, the starting position of $x_1$ and $x_2$ satisfies $a_2 - 8k \leq s +2\ell - 1 - 8\frac \ell 4 = s - 1$, while the end position of $y_1$ and $y_2$ satisfies $a_3 + 4k - 1 \geq s + \ell + 4\frac \ell 4 - 1 = s + 2\ell - 1$. 
Therefore, the square is entirely contained in the respective fragments corresponding to $x_1y_1$ and $x_2y_2$.
If $s < a_2 \leq s + 2\ell - 1$, we find the square with our choice of $x_1$ and $y_1$.
If $s < a_3 \leq s + 2\ell - 1$, we find the square with our choice of $x_2$ and $y_2$.
Otherwise, $T[s..s + 2\ell - 1]$ is entirely contained in tail $T[a_2..a_3]$, and we find another occurrence of the square further to the left.
\end{proof}

\subsection{Simple Algorithm for Detecting Squares}

Now we have all the tools to introduce our simple method for testing square-freeness of $T[1..n]$ using $\Oh(n (\log \sigma + \log \log n))$ comparisons, assuming that $\sigma$ is known in advance.
Let $\Delta = (\sigma\log n)^2$. We partition $T[1..n]$ into blocks of length $8\Delta$, and denote the $k^{\text{th}}$ block by $B_{k}$.
A square of length at most $8\Delta$ can be found by invoking \cref{lem:classical} on $B_{1}B_{2}$, $B_{2}B_{3}$, and so on. 
This takes $\Oh(\Delta \log \Delta) = \Oh(\Delta(\log \sigma + \log\log n))$ comparisons for each pair of adjacent blocks, or $\Oh(n(\log \sigma + \log\log n))$ comparisons in total. 
It remains to test for squares of length exceeding $8\Delta$.
This is done by
first invoking \cref{lem:compute} to compute a $\Delta$-approximate LZ factorisation of $T[1..n]$ with $\Oh(n\sigma\log n/\sqrt{\Delta}) = \Oh(n)$ comparisons, and then using \cref{lem:long}, which adds another $\Oh(n)$ comparisons. 
The total number of comparisons is dominated by the $\Oh(n(\log \sigma + \log \log n))$ comparisons needed to apply \cref{lem:classical} to the block pairs.

\subsection{Improved Algorithm for Detecting Squares}
\label{sec:improved}

We are now ready to describe the algorithm that uses only $\Oh(n\log \sigma)$ comparisons without knowing the 
value of $\sigma$. Intuitively, we will proceed in phases, trying to ``guess'' the value of $\sigma$. We first observe that
\cref{lem:compute} can be extended to obtain the following.

\begin{lemma}
\label{lem:compute2}
There is an algorithm that, given any parameter $\Delta \in [1,m]$, estimate $\tilde{\sigma}$ and fragment $T[x..y]$ of length $m$, uses
$\Oh(m\tilde{\sigma}\log m/\sqrt{\Delta})$ comparisons plus $\Oh(n)$ comparisons shared by all invocations of the lemma, and either computes a $\Delta$-approximate LZ factorisation of $T[x..y]$ or
determines that $\sigma>\tilde{\sigma}$.
\end{lemma}

\begin{proof}
We run the procedure described in the proof of \cref{lem:compute} and keep track of the number of comparisons
with negative answer. As soon as it exceeds $\Oh(m\tilde{\sigma}\log m/\sqrt{\Delta})$ (where the constant follows
from the complexity analysis) we know that necessarily $\sigma>\tilde{\sigma}$, so we can terminate. Otherwise, the algorithm
obtains a $\Delta$-approximate LZ factorisation with $\Oh(m\tilde{\sigma}\log m/\sqrt{\Delta})$ comparisons.
Comparisons with positive answer are paid for globally.
\end{proof}

Now we describe how to find any square using $\Oh(n \lg \sigma)$ comparisons.
We define the sequence $\sigma_t=2^{2^{\lceil\log\log n\rceil-t}}$, for $t=0,1,\ldots,\lceil\log\log n\rceil$.
We observe that $\sigma_{t-1}=(\sigma_{t})^{2}$, and proceed in phases corresponding to the values of $t$.
In the $t^{\text{th}}$ phase we are guaranteed that any square of length at least $(\sigma_{t})^{2}$ has been
already detected, and we aim to detect square of length less than $(\sigma_{t})^{2}$, and at least $\sigma_t$. We partition
the whole $T[1..n]$ into blocks of length $(\sigma_{t})^{2}$, and denote the $k^{\text{th}}$ block by $B_{k}$.
A square of length less than $(\sigma_{t})^{2}$ is fully contained within some two consecutive blocks $B_{i}B_{i+1}$,
hence we consider each such pair $B_{1}B_{2}$, $B_{2}B_{3}$, and so on.
We first apply \cref{lem:compute2} with $\Delta=\sigma_{t}/8$ and $\tilde{\sigma}=(\sigma_{t})^{1/4}/\log(\sigma_{t})$ to find
an $(\sigma_{t}/8)$-approximate LZ factorisation of the corresponding fragment of $T[1..n]$, and then
use \cref{lem:long} to detect squares of length at least $\sigma_t$.
We cannot always afford to apply \cref{lem:long} to all block pairs. Thus, we have to deactivate some of the blocks, which we explain when analysing the number of comparisons performed by the algorithm.
If any of the calls to \cref{lem:compute2} in the current phase detects that $\sigma>\tilde{\sigma}$, we switch to
applying \cref{lem:classical} on every pair of blocks $B_{i}B_{i+1}$ of the current phase and then terminate the whole algorithm.

We now analyse the total number of comparisons, ignoring the $\Oh(n)$ comparisons shared by all invocations of
\cref{lem:compute2}.
Throughout the $t^{\text{th}}$ phase, we use $\Oh(n \cdot \tilde{\sigma} \log \sigma_t / \sqrt{\Delta}) = \Oh(n\cdot (\sigma_{t})^{1/4}/\log(\sigma_{t})\cdot \log(\sigma_{t})/\sqrt{\sigma_{t}})=\Oh(n/(\sigma_{t})^{1/4})$ comparisons
to construct the $\Delta$-approximate factorisations (using \cref{lem:compute2})
until we either process all pairs of blocks or detect that $\sigma>(\sigma_{t})^{1/4}/\log(\sigma_{t})$. 
In the latter case, we finish off the whole computation with $\Oh(n\log(\sigma_{t}))$ comparisons (using \cref{lem:classical}), and by assumption on
$\sigma$ this is $\Oh(n\log \sigma)$ as required.
Until this happens (or until we reach phase $t=\ceil{\log \log n}-3$ where $\sigma_t \leq 256$), we use $\Oh(\sum_{t=0}^{t'} n/(\sigma_{t})^{1/4})$ comparisons
to construct the $\Delta$-approximate factorisations,
for some $0\leq t'\leq \lceil\log\log n\rceil$.
To analyse the sum, we need the following bound (made more general for the purpose of the next section).

\begin{lemma}
\label{lem:sum}
For any $0\leq x\leq y$ and $c \geq 0$ we have $\sum_{i=x}^{y} 2^{ic}/2^{2^{i}} = \Oh(2^{xc}/2^{2^{x}})$.
\end{lemma}

\begin{proof}
We observe that the sequence of exponents $2^{i}$ is strictly increasing from $i=0$, hence
\[ \sum_{i=x}^{y} \frac{2^{ic}}{2^{2^{i}}} \leq \sum_{i=2^{x}}^{2^{y}} \frac{i^c}{2^{i}} \leq \sum_{i=2^{x}}^{\infty} \frac{i^c}{2^{i}} = \sum_{i=0}^{\infty} \frac{(2^x+i)^c}{2^{(2^x+i)}} \leq \sum_{i=0}^{\infty} \frac{2^{xc} \cdot (i+1)^c}{2^{(2^x+i)}} =
\frac{2^{xc}}{2^{2^x}} \cdot \sum_{i=0}^{\infty} \frac{(i+1)^c}{2^{i}} .\]
$\sum_{i=0}^{\infty} \frac{(i+1)^c}{2^{i}}$ is a series of positive terms, we thus use Alembert's ratio test $ \frac{(i+2)^c}{2^{i+1}} \cdot \frac{2^{i}}{(i+1)^c} = \frac{1}{2}\frac{(i+2)^c}{(i+1)^c}$ which tends to $\frac{1}{2}$ when $i$ goes to the infinity, thus the series converges to a constant.
\end{proof}

\begin{corollary}
\label{lem:polylog}
For any $0 \leq t' \leq \ceil{\log \log n}$, it holds that $\sum_{t=0}^{t'} n \cdot \textnormal{polylog}(\sigma_t)/(\sigma_{t})^{1/4} = \Oh(n)$.
\begin{proof}\belowdisplayskip=-10pt
We have to show that $\sum_{t=0}^{t'} n \log^c(\sigma_t)/(\sigma_{t})^{1/4}) = \Oh(n)$ for any constant $c \geq 0$. We achieve this by splitting the sum and applying \cref{lem:sum}.
\begin{eqnarray*}
\sum_{t=0}^{t'} \frac{n \log^c(\sigma_t)}{(\sigma_{t})^{1/4}} &\leq& \sum_{t=0}^{\lceil \log\log n\rceil} \frac{n  \cdot (2^{\lceil \log\log n\rceil - t})^c}{(2^{2^{\lceil \log\log n\rceil - t}})^{1/4}} \enskip = \enskip \sum_{t=0}^{\lceil \log\log n\rceil} \frac{n  \cdot (2^{t})^c}{(2^{2^{t}})^{1/4}}\\
&=& n \cdot \sum_{t=0}^{\lceil \log\log n\rceil} \frac{2^{tc}}{2^{2^{t-2}}} \enskip=\enskip
n \cdot \left(\frac{1}{2^{2^{-2}}}+\frac{2^c}{2^{2^{-1}}} +  4 \cdot \sum_{t=0}^{\lceil \log\log n\rceil-2} \frac{2^{tc}}{2^{2^{t}}}\right) = \Oh(n)
\end{eqnarray*}
\null\hfill\null
\end{proof}
\end{corollary}

Thus, all invocations of \cref{lem:compute2} cause $\Oh(\sum_{t=0}^{t'} n/(\sigma_{t})^{1/4}) = \Oh(n)$ comparisons.

\subsubsection*{Deactivating Block Pairs}

It remains to analyse the number of comparisons used by \cref{lem:long} throughout all phases. As mentioned earlier, we cannot actually afford to apply \cref{lem:long} to all block pairs. Thus, we introduce a mechanism that deactivates some of the pairs.

First, note that there are $\Oh(\sum_{t=0}^{t'} n / (\sigma_t)^2) \subseteq \Oh(\sum_{t=0}^{t'} n / (\sigma_t)^{1/4}) = \Oh(n)$ block pairs in all phases. For each pair, we store whether it has been deactivated or not, where being deactivated broadly means that we do not have to investigate the pair because it does not contain a leftmost distinct square.
For each block pair $B_iB_{i + 1}$ in the current phase $t$, we first check if it has been marked as deactivated. If not, we also check if it has been \emph{implicitly} deactivated, i.e., if any of the two pairs from the previous phase that contain $B_iB_{i + 1}$ are marked as deactivated. If $B_iB_{i + 1}$ has been implicitly deactivated, then we mark it as deactivated and do not apply \cref{lem:compute2} and \cref{lem:long} (the implicit deactivation serves the purpose of propagating the deactivation to all later phases). Note that if some position of the string is not contained in any active block pair in some phase, then it is also not contained in any active block pair in all later phases. This is because it always holds that $\sigma_{t-1}=(\sigma_{t})^{2}$ (with no rounding required), which guarantees that block boundaries of earlier phases do not intersect blocks of later phases.

We only apply \cref{lem:compute2} and then \cref{lem:long} to $B_iB_{i + 1}$ if the pair has neither explicitly nor implicitly been deactivated. When applying \cref{lem:long}, a tail $T[a..a + \ell)$ contributes $\Oh(\ell)$ comparisons if $\ell \geq \Delta = \sigma_t / 8$ (and otherwise it contributes no comparisons).
As the whole $T[a..a + \ell)$ occurs earlier,  it cannot contain the leftmost occurrence of a square in the whole $T$. Thus,
any block pair (of any phase) contained in $T[a..a + \ell)$ also cannot contain such an occurrence, and thus such block pairs
can be deactivated.

The mechanism used for deactivation works as follows. Let $T[a..a + \ell)$ be a tail contributing $\Oh(\ell)$ comparisons with $\ell \geq \Delta = \sigma_t / 8$ in phase $t$. We mark all block pairs of phase $t + 2$ that are entirely contained in $T[a..a + \ell)$ as deactivated. Note that blocks in phase $t + 2$ are of length $\sqrt{\sigma_t}$, and consider the fragment $T[a + 2\sqrt{\sigma_t}..a + \ell - 2\sqrt{\sigma_t})$. In phase $t + 2$, and by implicit deactivation in all later phases, this fragment overlaps (either partially or fully) only block pairs that have been deactivated. Thus, after phase $t + 1$, we will never inspect any of the symbols in $T[a + 2\sqrt{\sigma_t}..a + \ell - 2\sqrt{\sigma_t})$ again. We say that tail $T[a..a + \ell)$ deactivated the fragment of length $\ell - 4\sqrt{\sigma_t} = \Omega(\ell)$, which is positive until phase $t=\ceil{\log \log n}-3$ because $\sigma_t > 256$. Since the number of deactivated positions is linear in the number of comparisons that the tail contributes to \cref{lem:long}, it suffices to show that each position gets deactivated at most a constant number of times. 
In a single phase, any position gets deactivated at most twice. This is because the tails of each factorisation do not overlap by definition, but the tails of the two factorisations of adjacent block pairs $B_iB_{i + 1}$ and $B_{i + 1}B_{i + 2}$ can overlap. If a position gets deactivated for the first time in phase $t$, then (as explained earlier) we will not consider it in any of the phases $t' \geq t + 2$. Thus, it can only be that we deactivate the position again in phase $t + 1$, but not in any later phases. In total, each position gets deactivated at most four times. Hence \cref{lem:long} contributes $\Oh(n)$ comparisons in total.

We have shown:

\upperbound*

\section{Algorithm}
\label{sec:alg}

In this section we show how to implement the approach described in the previous section
to work in $\Oh(n\log \sigma)$ time. The main
difficulty is in efficiently implementing the sparse suffix tree construction algorithm, and
then computing a $\Delta$-approximate factorisation. We first how to obtain an $\Oh(n\log \sigma+n\log^{*}n)$ time
algorithm that still uses only $\Oh(n\log \sigma)$ comparisons, and then further improve its running time to $\Oh(n\log \sigma)$.

\subsection[Constructing Suffix Tree and Delta-Approximate Factorisation]{Constructing Suffix Tree and \boldmath$\Delta$\unboldmath-Approximate Factorisation}
\label{sec:suffixtree}

To give an efficient algorithmic construction of the sparse suffix tree from \cref{lem:sparse}, we will use a restricted version of LCEs, where a query $\textnormal{ShortLCE}_x(i, j)$ (for any positive integer $x$) returns $\min(x, \lce(i, j))$. The following result was given by Gawrychowski et al.~\cite{Gawrychowski2016}:

\begin{lemma}[Lemma 14 in~\cite{Gawrychowski2016}]
\label{lem:LCE_undordered_original}
For a length-$n$ string over a general unordered alphabet\footnote{Lemma 14 in~\cite{Gawrychowski2016} does not explicitly mention that it works for general unordered alphabet. However, the proof of the lemma relies solely on equality tests.}, a sequence of $q$ queries $\textnormal{ShortLCE}_{4^{k_i}}$ for $i \in \{1, \dots, q\}$ can be answered online in total time $\Oh(n \log^* n + s)$ and $\Oh(n + q)$ comparisons\footnote{Lemma 14 in~\cite{Gawrychowski2016} does not explicitly mention that it requires $\Oh(n + q)$ comparisons. However, they use a union-find approach where there can be at most $\Oh(n)$ comparisons with outcome "equal", and each LCE query performs only one comparison with outcome "not-equal", similarly to what we describe in the proof of \cref{lem:sparse}.}, where $s = \sum_{i = 1}^q(k_i + 1)$.
\end{lemma}

In the lemma, apart from the $\Oh(n \log^* n)$ time, each $\textnormal{ShortLCE}_{4^{k_i}}$ query accounts for $\Oh(k_i + 1)$ time. Note that we can answer the queries online, without prior knowledge of the number and length of the queries. Also, computing an LCE in a fragment $T[x..y]$ of length $m$ trivially reduces to a $\textnormal{ShortLCE}_{4^{\ceil{\log_4 m}}}$ query on $T$. Thus, we have:

\begin{corollary}
\label{lem:LCE_undordered}
A sequence of $q$ longest common extension queries on a fragment $T[x..y]$ of length $m$ over a general unordered alphabet can be answered in $\Oh(q \log m)$ time plus $\Oh(n \log^* n)$ time shared by all invocations of the lemma. The number of comparisons is $\Oh(q)$, plus $\Oh(n)$ comparisons shared by all invocations of the lemma.
\end{corollary}

While constructing the sparse suffix tree, we will maintain a heavy-light decomposition using a rebuilding scheme introduced by Gabow~\cite{Gabow1990}. Let $L(u)$ denote the number of leaves in the subtree of a node $u$.
We use the following recursive construction of a heavy-light decomposition: Starting from a node $r$ (initially the root of the tree), we find the deepest descendant node $e$ such that $L(e)\geq \frac{5}{6}L(r)$ (possibly $e = r$). The path $p$ from the root $r(p)=r$ to $e(p)=e$ is a heavy path. Any edge $(u,v)$ on this path satisfies $L(v)\geq \frac{5}{6}L(u)$, and we call those edges \emph{heavy}. As a consequence, a node $u$ can have at most one child $v$ such that $(u,v)$ is heavy. For each edge $(u,v)$ where $u$ is on the heavy path and $v$ is not, we recursively build a new heavy path construction starting from~$v$.

When inserting a new suffix in our tree, we keep track of the insertion in the following way: for every root of a heavy path, we maintain $I(u)$ the number of insertions made in the subtree of $u$ since we built the heavy-light decomposition of this subtree. When $I(u) \geq \frac{1}{6} L(u)$ we recalculate the values of $L(v)$ for all nodes $v$ in the subtree of $u$ and rebuild the heavy-light decomposition for the subtree of $u$.

This insures that, despite insertion, for any heavy path starting at node $r$ and 
a node $u$ on that heavy path, $L(e) \geq \frac{2}{3} L(r)$. When crossing a non-heavy edge the number of nodes in the subtree reduces by a factor at least $\frac{5}{6}$ which leads to the following property:

\begin{observation}
\label{obs:heavy}
The path from any node to the root crosses at most $\Oh(\log m)$ heavy paths.
\end{observation}

Additionally, rebuilding a subtree of size $s$ takes $\Oh(s)$ time and adding a suffix $T[i_j .. y]$ to the tree increases $I(r)$ for each path $p$ from the root $r$ to the new leaf. Those are at most $\Oh(\log m)$ nodes, and thus maintaining the heavy path decomposition takes amortized time $\Oh(\log n)$ time per insertion.

With these building blocks now clearly defined, we are ready to describe the construction of the sparse suffix tree.

\begin{lemma}
\label{lem:sparse-algo}
The sparse suffix tree containing any $b$ suffixes $T[i_{1}..y], \ldots, T[i_{b}..y]$
of $T[x..y]$ with $m=|T[x..y]|$ can be constructed using $\Oh(b\sigma\log b\log m)$ time plus $\Oh(n \log^* n)$ time shared by all invocations of the lemma.
\end{lemma}

\begin{proof}
As in the proof of \cref{lem:sparse}, we consider the insertion of a suffix $T[i_j .. y]$ into the sparse suffix tree with suffixes $T[i_{1} .. y], T[i_{2} .. y] $ $\cdots T[i_{j-1} .. y]$.
At all times, we maintain the heavy path decomposition. Additionally, we maintain for each heavy path a predecessor data structure, where given some length $\ell$, we can quickly identify the deepest explicit node on the heavy path that spells a string of length at least $\ell$. The data structure can, e.g., be a balanced binary search tree with insertion and search operations in $\Oh(\log b)$ time (the final sparse suffix tree and thus each heavy path contains $\Oh(b)$ nodes). 
When rebuilding a subtree of the heavy path decomposition, we also have to rebuild the predecessor data structure for each of its heavy paths. 
Thus, rebuilding a size-$q$ subtree takes $\Oh(q \log b)$ time (each node is on exactly one heavy path and has to be inserted into one predecessor data structure), and the amortized insertion time increases from $\Oh(\log m)$ to $\Oh(\log m \cdot \log b)$.
Whenever we insert a suffix, we make at most one node explicit, and thus have to perform at most one insertion into a predecessor data structure. The time for this is $\Oh(\log b)$, which is dominated be the previous term.

When inserting $T[i_j .. y]$, we look for the node $u$ corresponding to the longest common prefix between $T[i_j .. y]$ and the inserted suffixes, make $u$ explicit if necessary and add a new leaf corresponding to $T[i_j .. y]$ attached to $u$.
Let $v$ be the current node (initialized by the root, and always an explicit node) and $v_1, \cdots, v_d$ be its (explicit) children. If there is a heavy edge $(v, v_a)$ for $1 \leq a \leq d$, let $p$ be the corresponding heavy path.
For each heavy path $p$, we store the label of one leaf (i.e., the starting position of one suffix) that is contained in the subtree of $e(p)$. 
Thus, we can use \cref{lem:LCE_undordered} to
compute the longest common extension between the string spelled by $e(p)$ and $T[i_j .. y]$.
Now we use the predecessor data structure on the heavy path to find the deepest (either explicit or implicit) node $v'$ on the path that spells a prefix of $T[i_j .. y]$. If $v'$ is implicit, we make it explicit and add the leaf. If $v'$ is explicit and $v' \neq v$, we use $v'$ as the new current node and continue. Otherwise, we have $v' = v$, i.e., the suffix does not belong to the subtree rooted in $v_a$. In this case, we issue $d$ LCE queries between $T[i_j .. y]$ and each of the strings spelled by the nodes $v_1, \dots, v_d$. This either reveals that we can continue using one of the $v_a$ as the new current node, or that we can create a new explicit node on some $(v, v_a)$ edge and attach the leaf to it, or that we can simply attach a new leaf to $v$.

Now we analyse the time spent while inserting one suffix. We spent $\Oh(b \cdot \log m \cdot \log b)$ total time for inserting $\Oh(b)$ nodes into the dynamic heavy path decomposition and the predecessor data structures. In each step of the insertion process, we either (i) move as far as possible along some heavy path or (ii) move along some non-heavy edge. For (i), we issue one LCE query and one predecessor query. For (ii) we issue $\Oh(\sigma)$ LCE queries.
Due to \cref{obs:heavy}, both (i) and (ii) happen at most $\Oh(\log b)$ times per suffix. Thus, for all suffixes, we perform $\Oh(b \log b)$ predecessor queries and $\Oh(b\sigma \log b)$ LCE queries. The total time is $\Oh(b \log^2 b)$ for predecessor queries, and $\Oh(b\sigma \log b \log m)$ for LCE queries (apart from the $n \log^* n$ time shared by all invocations of \cref{lem:LCE_undordered}).
\end{proof}

\newcommand{\maxleft}{\textnormal{\textsf{src}}}
\newcommand{\maxleftlce}{\textnormal{\textsf{len}}}
\newcommand{\phraseend}{\textnormal{\textsf{phraseEnd}}}

\begin{lemma}
\label{lem:lz-log-star}
For any parameter $\Delta \in [1,m]$, a $\Delta$-approximate LZ factorisation of any fragment $T[x..y]$ of length $m$ can be
computed in $\Oh(m\sigma\log^2 m / \sqrt{\Delta})$ time plus $\Oh(n \log^* n)$ time shared by all invocations of the lemma.
\end{lemma}

\begin{proof}
Let $T' = T[x..y]$, and let $\{i_1, i_2, \dots,i_b\}$ be a $\Delta$-cover of $\{1, \dots, m\}$, which implies $b = \Theta(m / \sqrt{\Delta})$. We obtain a sparse suffix tree of the suffixes $T'[i_{1}..m],\ldots,T'[i_{b}..m]$, which takes $\Oh(b\sigma\log b\log m) \subseteq \Oh(m\sigma\log^2 m / \sqrt{\Delta})$ time according to \cref{lem:sparse-algo}, plus $\Oh(n \log^* n)$ time shared by all invocations of the lemma. Now we compute a $\Delta$-approximate LZ factorisation of $T'$ from the spare suffix tree in $\Oh(b)$ time.  

In the following proof, we use $i_1, i_2, \dots,i_b$ interchangeably to denote both the difference cover positions, as well as their corresponding leaves in the sparse suffix tree. 
Assume that the order of difference cover positions is $i_1 < i_2 < \cdots < i_b$. 
First, we determine for each $i_k > i_1$, the position $\maxleft(i_k) = i_h$ and the length $\maxleftlce(i_k) = \lce(i_h, i_k)$, where $i_h \in \{i_1, \dots, i_{k - 1}\}$ is a position that maximizes $\lce(i_h, i_k)$.
This is similar to what was done in \cite{Fischer2018} for the LZ77 factorisation.
We start by assigning labels from $\{1, \dots, b\}$ to the nodes of the sparse suffix tree. 
A node has label $k$ if and only of $i_k$ is its smallest descendant leaf.
We assign the labels as follows. Initially, all nodes are unlabelled. We assign label 1 to each node on the path from $i_1$ to the root.
Then, we process the remaining leaves $i_2, \dots,i_b$ in increasing order.
For each $i_k$, we follow the path from $i_k$ to the root. 
We assign label $k$ to each unlabelled node that we encounter. As soon as we reach a node that has already been labelled, say, with label $h$ and string-depth $\ell$%
, we are done processing leaf $i_k$. It should be easy to see that $i_h$ is also exactly the desired index that maximizes $\lce(i_h, i_k)$, and we have $\lce(i_h, i_k) = \ell$. Thus, we have found $\maxleft(i_k) = i_h$ and $\maxleftlce(i_k) = \ell$. 
The total time needed is linear in the number of sparse suffix tree nodes, which is $\Oh(b)$.

Finally, we obtain a $\Delta$-approximate LZ factorisation using $\maxleft$ and $\maxleftlce$. The previously computed values can be interpreted as follows: $i_k$ could become the starting position of a length-$\maxleftlce(i_k)$ tail (with previous occurrence at position $\maxleft(i_k)$). 
For the $\Delta$-approximate LZ factorisation, we will create the factors greedily in a left-to-right manner. 
Assume that we already factorised $T'[1..s - 1]$, then the next phrase starts at position $s$, and thus the next tail starts within $T'[s..s + \Delta)$ (as a reminder, the head is by definition shorter than $\Delta$). 
Let $S = \{i_1, i_2, \dots,i_b\} \cap \{s, \dots, s + \Delta - 1\}$. 
If there is no $i_k \in S$ with $i_k + \maxleftlce(i_k) > s + \Delta - 1$, then the next phrase is simply $T'[s..\min(|T'|,s + \Delta - 1))$ with empty tail. 
Otherwise, the next phrase has (possibly empty) head $T'[s..i_k)$ and tail $T'[i_k..i_k + \maxleftlce(i_k))$ (with previous occurrence $\maxleft(i_k)$), where $i_k$ is chosen from $S$ such that it maximizes $i_k + \maxleftlce(i_k)$.
Creating the phrase in this way clearly takes $\Oh(\absolute{S})$ time.
Since the next phrase starts at least at position $s + \Delta - 1$, none of the positions from $S \setminus \{s + \Delta - 1\}$ will ever be considered as starting positions of other tails. Thus, every $i_k$ is considered during the creation of at most two phrases, and the total time needed to create all phrases is $\Oh(b)$.

\newcommand{\iright}{i_k}
\newcommand{\ileft}{i_{k'}}

It remains to be shown that the computed factorisation is indeed a $\Delta$-approximate LZ factorisation, i.e., if we output a phrase $T'[s..e]$, then the unique (non-approximate) LZ phrase $T'[s..e']$ starting at position $s$ satisfies $e' - 1 \leq e$. 
First, note that for the created approximate phrases (except possibly the last phrase of $T$) we have $s + \Delta - 2 \leq e$.
Assume $e' < s + \Delta$, then clearly $e' - 1 \leq e$.
Thus, we only have to consider $e' > s + \Delta - 1$.
Since $T'[s..e']$ is an LZ phrase, there is some $s' < s$ such that $\lce(s', s) = e' - s$.
Let $h$ be the constant-time computable function that defines the $\Delta$-cover, and let $\ileft = s' + h(s', s)$ and $\iright = s + h(s', s)$. 
Note that $\ileft \in \{i_1, i_2, \dots,i_{k - 1}\}$ and $\iright \in \{i_1, i_2, \dots,i_b\} \cap \{s, \dots, s + \Delta - 1\}$.
Therefore, we have $\maxleftlce(\iright) \geq \lce(\ileft, \iright) = \lce(s', s) - h(s', s) = (e' - s) - (\iright - s) = e' - \iright$. While computing the $\Delta$-approximate phrase $T'[s..e]$, we considered $\iright$ as the starting positions of the tail, which implies $e \geq i_k + \maxleftlce(i_k) - 1 \geq e' - 1$.
\end{proof}

\begin{lemma}
\label{lem:compute3}
There is an algorithm that, given any parameter $\Delta \in [1,m]$, estimate $\tilde{\sigma}$ and fragment $T[x..y]$ of length $m$, takes
$\Oh(m\tilde{\sigma}\log^2 m/\sqrt{\Delta})$ time plus $\Oh(n \log^* n)$ time shared by all invocations of the lemma, and either computes a $\Delta$-approximate LZ factorisation of $T[x..y]$ or
determines $\sigma>\tilde{\sigma}$.
\begin{proof}
We simply use \cref{lem:lz-log-star} to compute the factorisation. In the first step, we have to construct the sparse suffix tree using the algorithm from \cref{lem:sparse-algo}. While this algorithm takes $\Oh(m\sigma\log^2 m / \sqrt{\Delta})$ time, it is easy to see that a more accurate time bound is $\Oh(md\log^2 m / \sqrt{\Delta})$, where $d$ is the maximum degree of any node in the sparse suffix tree. If during construction the maximum degree of a node becomes $\tilde{\sigma} + 1$, we immediately stop and return that $\sigma>\tilde{\sigma}$. Otherwise, we finish the construction in the desired time.
\end{proof}
\end{lemma}

Now we can describe the algorithm that detects squares in $\Oh(n \lg \sigma + n \log^* n)$ time and $\Oh(n \lg \sigma)$ comparisons. 
We simply use the algorithm from \cref{sec:upper}, but use \cref{lem:compute3} instead of \cref{lem:compute2}. Next, we analyse the time needed apart from the $\Oh(n \log^* n)$ time shared by all invocations of \cref{lem:compute3}. Throughout the $t^{\text{th}}$ phase, we use $\Oh(n \cdot \tilde{\sigma} \cdot \log^2(\sigma_t) / \sqrt{\Delta}) = \Oh(n\cdot (\sigma_{t})^{1/4}/\log(\sigma_{t})\cdot \log^2(\sigma_{t})/\sqrt{\sigma_{t}})=\Oh(n \log (\sigma_t)/(\sigma_{t})^{1/4})$ comparisons
to construct all the $\Delta$-approximate factorisations. As before, if at any time we discover that $\sigmaapprox > (\sigma_t)^{1/4} / \log(\sigma_t)$, then we use \cref{lem:classical} to finish the computation in $\Oh(n \lg \sigma_t) = \Oh(n \log \sigma)$ time. Until then (or until we finished all $\ceil{\log \log n}$ phases), we use $\Oh(\sum_{t=0}^{t'} n \log(\sigma_t)/(\sigma_{t})^{1/4})$ time, and by \cref{lem:polylog} this is $\Oh(n)$. 
For detecting squares, we still use \cref{lem:long}, which as explained in \cref{sec:upper} takes $\Oh(n)$ time and comparisons in total, plus additional $\Oh(Z)$ time, where $Z$ is the number of approximate LZ factors considered during all invocations of the lemma. 
We apply the lemma to each approximate LZ factorisation exactly once, and by construction each factor in phase $t$ has size at least $\Delta = \Omega(\sigma_t)$. Also, each text position is covered by at most two tails per phase. Hence $Z = \Oh(\sum_{t=0}^{t'} n/\sigma_t)$, which is $\Oh(n)$ by \cref{lem:polylog}.

The last thing that remains to be shown is how to implement the bookkeeping of blocks, i.e., in each phase we have to efficiently deactivate block pairs as described at the end of \cref{sec:upper}. 
We maintain the block pairs in $\ceil{\log \log n}$ bitvectors of total length $\Oh(n)$, where a set bit means that a block pair has been deactivated (recall that there are $\Oh(n)$ pairs in total). 
Bitvector $t$ contains at position $j$ the bit corresponding to block pair $B_{j}B_{j + 1} = T[i..i+2(\sigma_t)^2)$ with $i = (\sigma_t)^2 \cdot (j - 1)$. Note that translating between $i$ and $j$ takes constant time.
For each sufficiently long tail in phase $t$, we simply iterate over the relevant block pairs in phase $t + 2$ and deactivate them, i.e., we set the corresponding bit. This takes time linear in the number of deactivated blocks.
Since there are $\Oh(n)$ block pairs, and each block pair gets deactivated at most a constant number of times, the total cost for this bookkeeping is $\Oh(n)$.

The number of comparisons is dominated by the $\Oh(n \log \sigma)$ comparisons used when finishing the computation with \cref{lem:classical}. The only other comparisons are performed by \cref{lem:long}, which we already bounded by $\Oh(n)$, and by LCE queries via \cref{lem:LCE_undordered}. Since we ask $\Oh(n)$ such queries in total, the number of comparisons is also $\Oh(n)$.
We have shown:

\begin{lemma}
The square detection algorithm from \cref{sec:upper} can be implemented in $\Oh(n \lg \sigma + n \log^* n)$ time and $\Oh(n \log \sigma)$ comparisons.
\end{lemma}

\subsection{Final Improvement}
\label{sec:finalimprov}

\definecolor{new-red-fill}{HTML}{f1a340}
\colorlet{new-red-line}{new-red-fill!50!white}
\colorlet{new-red-line-end}{black}
\tikzset{redonlyhatch/.style={
	pattern=north west lines, 
	pattern color=new-red-line},
}
\tikzset{redhatch/.style={
	preaction={fill, new-red-fill}, 
	redonlyhatch},
}

\definecolor{new-blue-fill}{HTML}{998ec3}
\colorlet{new-blue-line}{new-blue-fill!75!black}
\colorlet{new-blue-line-end}{black}
\tikzset{blueonlyhatch/.style={
	pattern=north east lines, 
	pattern color=new-blue-line,
}}

\tikzset{bluehatch/.style={
	preaction={fill, new-blue-fill}, 
	blueonlyhatch,
}}

\colorlet{new-dense-fill}{black!10!white}
\colorlet{new-dense-line}{black!40!white}
\colorlet{new-dense-line-end}{black}
\tikzset{densehatch/.style={
	preaction={fill, new-dense-fill}, 
	pattern color=new-dense-line},
}

\colorlet{new-match-fill}{white}
\colorlet{new-match-line}{white}
\colorlet{new-match-line-end}{black}
\tikzset{matchhatch/.style={
	preaction={fill, new-match-fill}, 
	pattern=crosshatch, 
	pattern color=new-match-line},
}

\begin{figure}

\centering

\subcaptionbox{\label{fig:improv:1}Sampling dense fragments and cutting the text into chunks. Dotted lines indicate chunk boundaries, and $h_x = (j + x)\cdot \tau$ for some integer $j$ and $x \in [0, 13]$ are positions of chunk boundaries. 
The dense fragments are $D_1 = T[h_2..h_3)$, $D_2 = T[h_7..h_8)$, and $D_3 = T[h_{12}..h_{13})$. 
The primary occurrences of dense fragments are grey, while the secondary occurrences (the ones that we aim to find) are white.
A purple box in the text, and the matching purple line underneath the text, correspond to some substring $T[j\cdot\tau-r_{j - 1}..j\cdot\tau)$. Similarly, the orange boxes and lines correspond to substrings $T[j\cdot\tau..j\cdot\tau + \ell_j)$.
}{
\begin{tikzpicture}[x=.0098\textwidth, y=1.3em]

\tikzset{slimfit/.style={inner sep=0, outer sep=0, minimum width=0, minimum height=0}}
\tikzset{every node/.style={slimfit}}

\foreach[count=\xplus from 5] \x in {4,...,99} {
	\node (\x-tl) at (\x,0) {};
	\node (\x-br) at (\xplus,1) {};
	\node[fit=(\x-tl)(\x-br)] (\x) {};
	\node (\x-tl2) at (\x,0) {};
	\node (\x-br2) at (\xplus,1) {};
	\node[fit=(\x-tl2)(\x-br2)] (\x-2) {};
}

\node[fit=(4)(9)] (prefix) {};
\node[fit=(94)(99)] (suffix) {};

\foreach \x/\y in {%
suffix.north east/suffix.north west,%
prefix.north east/prefix.north west,%
suffix.south east/suffix.south west,%
prefix.south east/prefix.south west,%
suffix.north east/suffix.south east,%
prefix.north west/prefix.south west%
}
{
\draw[densely dotted, thick] (\x) to (\y);
}

\foreach[count=\i from 1,
count=\xprime from 2, %
evaluate=\xprime as \x using int(2+5*\xprime),%
evaluate=\x as \y using int(\x-\z+1)] \z in {1,3,1,5,3,3,0,1,5,3,0,2,4,5,2,4} {
\node[bluehatch, fit=(\x.north east)(\y.south west)] (pattern) {};
\draw[new-blue-line-end] (pattern.north west) to (pattern.south west);

\path (pattern.south west) ++(0.15,-0.3) node (\i-ptl) {};
\path (pattern.east |- \i-ptl) node (\i-ptc) {};
}

\foreach[count=\i from 1,
count=\xprime from 2, %
evaluate=\xprime as \x using int(3+5*\xprime),%
evaluate=\x as \y using int(\x+\z-1)] \z in {1,3,5,1,2,4,2,5,0,3,0,0,5,1,3,4} {
\node[redhatch, fit=(\x.north west)(\y.south east)] (pattern) {};
\draw[new-red-line-end] (pattern.north east) to (pattern.south east);

\path (pattern.east |- \i-ptl) ++(-0.15,-0.2) node (\i-pbr) {};
\ifnum\z=0
\path (\i-pbr |- \i-ptc) node (\i-ptc) {};
\fi
}

\foreach[count=\i from 1, evaluate=\multi as \offset using -\multi*0.35] \multi in {
0,
0,
0,
1,
0,
0,
0,
0,
1,
0,
42,
0,
0,
1,
0,
1
} {
\ifnum\multi<42
\path (\i-ptl) ++(0, \offset) node (\i-ptl) {};
\path (\i-ptc) ++(0, \offset) node (\i-ptc) {};
\path (\i-pbr) ++(0, \offset) node (\i-pbr) {};
\node[bluehatch, fit=(\i-ptl)(\i-ptc |- \i-pbr)] (lbox) {};
\node[redhatch, fit=(\i-ptc)(\i-pbr)] (rbox) {};
\draw[new-blue-fill] (lbox.north west) to (lbox.south west);
\draw[new-blue-fill] (lbox.north west) to (lbox.north east);
\draw[new-red-fill] (rbox.north east) to (rbox.south east);
\draw[new-red-fill] (rbox.north west) to (rbox.north east);
\draw[thick] (lbox.south west) to (rbox.south east);
\fi
}

\node[fit=(83)(84), shading=axis, left color=new-red-fill, right color=new-red-fill!50!new-blue-fill] {};
\node[fit=(85)(87), shading=axis, left color=new-red-fill!50!new-blue-fill, right color=new-blue-fill, pattern=vertical lines] {};
\filldraw[blueonlyhatch, draw=none] (87.north east) -- (86.north west) -- 
(84.south west) -- (87.south east);
\filldraw[redonlyhatch, draw=none] (83.north west) -- (86.north west) -- 
(84.south west) -- (83.south west);


\foreach[count=\i from 1, evaluate=\x as \y using int(\x+4)] \x in {23,48,73} {
\draw[densely dashed] (\x.north west) ++(0,0) to node[pos=1.0, below] (l\i) {} ++(0,-2.5);
\node[fit=(\x)(\y), densehatch, draw=new-dense-line-end] (d\i) {};
}

\node at (d1.center) {\contourlength{0pt}\contour{new-dense-fill}{$D_1$}};
\node at (d2.center) {\contourlength{0pt}\contour{new-dense-fill}{$D_2$}};
\node at (d3.center) {\contourlength{0pt}\contour{new-dense-fill}{$\null\enskip\, D_3$}};

\foreach[count=\i from 1] \x in {$i\cdot \tau^{2}$,$(i+1)\cdot \tau^{2}$,$(i+2)\cdot \tau^{2}$} {
\node[below=0.2em of l\i] {$\strut$\x};
}

\draw[thick] (prefix.north east) to (suffix.north west);
\draw[thick] (prefix.south east) to (suffix.south west);

\node[left=0 of prefix] {$\phantom{T'}\mathllap{T} =\enskip$};

\foreach[count=\i from 1] \x in {13,18,...,90} {
\draw[densely dotted, thick] (\x.north west) ++(0,0.5) to node[pos=0, outer sep=.1em] (l\i) {} ++(0,-1.5);
}

\foreach[count=\i from 0] \j in {1,2,3,4,5,6,7,8,9,10,11,12,13,14} {
\node[above=.25em of l\j] {$\scriptstyle h_{\i}$};
}

\node[fit=(36-2)(40-2), matchhatch, draw=new-match-line-end] (d) {};
\node at (d.center) {$D_3$};
\node[fit=(55-2)(59-2), matchhatch, draw=new-match-line-end] (d) {};
\node at (d.center) {$D_3$};
\node[fit=(70-2)(74-2), matchhatch, draw=new-match-line-end] (d) {};
\node at (d.center) {$D_1$};

\draw[densely dashed, gray] (11.south) ++(0, -1) to ++(0,3) to[out=90,in=270] ++(-5em,3em) to node[pos=1] (left) {} ++(0,3em);

\draw[densely dashed, gray] (21.south) ++(0, -1) to ++(0,3) to[out=90,in=270,looseness=.6] ++(10em,3em) to node[pos=1] (right) {} ++(0,3em);

\draw (right) ++(.5em,-1) node (rright) {};
\draw (rright) ++(0,1) node (rrightup) {};

\draw (left) ++(-.5em,-1) node (lleft) {};
\draw (lleft) ++(0,1) node (lleftup) {};

\node[below=.5em of lleft] (down) {};

\path (left) to node[pos=.05] (b1) {} (right);
\path (left) to node[pos=.15] (d1) {} (right);
\path (left) to node[pos=.25] (r1) {} (right);
\path (left) to node[pos=.35] (b2) {} (right);
\path (left) to node[pos=.65] (d2) {} (right);
\path (left) to node[pos=.95] (r2) {} (right);

\node[bluehatch, fit=(b1 |- lleftup)(d1 |- lleft)] (pattern) {};
\draw[new-blue-line-end] (pattern.north west) to (pattern.south west);
\node[fit=(pattern), inner sep=-.1em] (pattern) {};
\draw[<->] (pattern.west |- down) to node[midway, below=.2em] {$\strut r_{j-1}\enskip\null$} (pattern.east |- down);

\node[bluehatch, fit=(b2 |- lleftup)(d2 |- lleft)] (pattern) {};
\draw[new-blue-line-end] (pattern.north west) to (pattern.south west);
\node[fit=(pattern), inner sep=-.1em] (pattern) {};
\draw[<->] (pattern.west |- down) to node[midway, below=.2em] {$\strut r_{j}$} (pattern.east |- down);

\node[redhatch, fit=(r1 |- lleftup)(d1 |- lleft)] (pattern) {};
\draw[new-red-line-end] (pattern.north east) to (pattern.south east);
\node[fit=(pattern), inner sep=-.1em] (pattern) {};
\draw[<->] (pattern.west |- down) to node[midway, below=.2em] {$\strut \ell_{j}$} (pattern.east |- down);

\node[redhatch, fit=(r2 |- lleftup)(d2 |- lleft)] (pattern) {};
\draw[new-red-line-end] (pattern.north east) to (pattern.south east);
\node[fit=(pattern), inner sep=-.1em] (pattern) {};
\draw[<->] (pattern.west |- down) to node[midway, below=.2em] {$\strut \ell_{j+1}$} (pattern.east |- down);

\draw[densely dotted, thick] (d1 |- lleft) to node[pos=.75, outer sep=.2em] (l1) {} ++(0,2);
\draw[densely dotted, thick] (d2 |- lleft) to node[pos=.75, outer sep=.2em] (l2) {} ++(0,2);

\draw[<->] (l1) to node[midway, above=.2em] {$\tau$} (l2);

\node[above=1em of l1] {$\strut h_0 = j\cdot\tau$};
\node[above=1em of l2] {$\strut h_1 = (j+1)\cdot\tau$};

\draw[thick] (lleft) to (rright);
\draw[thick] (lleftup) to (rrightup);

\end{tikzpicture}

\vspace{.5\baselineskip}
}

\vspace{2\baselineskip}

\subcaptionbox{\label{fig:improv:2}The string $T'$ used to find all the occurrences of dense fragments. Each position $\hat{h}_x$ maps to position $h_x$ in \cref{fig:improv:1}. The substring indicated by the purple box preceding $h_x = (j+x)\tau$ and the orange box succeding $h_x$ is exactly $T[h_x - r_{j + x - 1}..h_x + \ell_{j + x})$. Each $\underset{\smash{\scriptstyle x}}{\smash{\texttt{\textdollar}}}%
$ is a distinct separator symbol that is unique within $T'$.}{
\begin{tikzpicture}[x=.0098\textwidth, y=1.3em]

\tikzset{slimfit/.style={inner sep=0, outer sep=0, minimum width=0, minimum height=0}}
\tikzset{every node/.style={slimfit}}

\foreach[count=\xplus from 5] \x in {4,...,99} {
	\node (\x-tl) at (\x,0) {};
	\node (\x-br) at (\xplus,1) {};
	\node[fit=(\x-tl)(\x-br)] (\x) {};
	\node (\x-tl2) at (\x,-1) {};
	\node (\x-br2) at (\xplus,0) {};
	\node[fit=(\x-tl2)(\x-br2)] (\x-2) {};
	\node (\x-tl3) at (\x,-1) {};
	\node (\x-br3) at (\xplus,0) {};
	\node[fit=(\x-tl3)(\x-br3)] (\x-3) {};
}

\node[fit=(4)(6)] (prefix) {};
\node[fit=(97)(99)] (suffix) {};

\foreach \x/\y in {%
suffix.north east/suffix.north west,%
prefix.north east/prefix.north west,%
suffix.south east/suffix.south west,%
prefix.south east/prefix.south west,%
suffix.north east/suffix.south east,%
prefix.north west/prefix.south west%
}
{
\draw[densely dotted, thick] (\x) to (\y);
}

\foreach[
remember=\nextfirst as \first (initially 8),
evaluate=\first as \nextfirst using int(\first+\a+\b+2),
evaluate=\first as \last using int(\first+\a+\b-1),
evaluate=\first as \mid using int(\first+\a),
evaluate=\first as \matchA using int(\first+\x),
evaluate=\first as \matchAend using int(\first+\x+4),
evaluate=\first as \matchB using int(\first+\n),
evaluate=\first as \matchBend using int(\first+\n+4),
count=\i from 0
] \a/\b/\x/\y/\n/\m in {%
1/1/0/0/0/0,
3/3/0/0/0/0,
1/5/1/1/0/0,
5/1/0/1/0/0,
3/2/0/0/0/0,
3/4/0/0/1/3,
0/2/0/0/0/0,
1/5/1/2/0/0,
5/0/0/2/0/0,
3/3/0/0/0/3,
0/0/0/0/0/0,
2/0/0/0/0/0,
4/5/4/3/1/1
}{
\node[bluehatch, fit=(\first.north west)(\mid.south west)] (pattern) {};
\draw[new-blue-line-end] (pattern.north west) to (pattern.south west);
\node[redhatch, fit=(\mid.north west)(\last.south east)] (pattern) {};
\draw[new-red-line-end] (pattern.north east) to (pattern.south east);
\draw[densely dotted, thick] (\mid.north west) ++(0,0.5) to node[pos=0, outer sep=.1em] (l\i) {} ++(0,-1.5);
\node[above=.25em of l\i] {$\scriptstyle \hat{h}_{\i}$};

\node[fit=(\last.north east)(\nextfirst.south west)] (dollar) {};
\node at (dollar.center) {%
\raisebox{-5.5pt}{%
$\underset{\smash{\scriptscriptstyle\i}}{\smash{\scriptstyle\texttt{\textdollar}}}%
$%
}};

\ifnum\y>0
\node[fit=(\matchA-2)(\matchAend-2), densehatch, draw=new-dense-line-end] (dense) {};
\ifnum\i=12
\node at (dense.center) {\contourlength{0pt}\contour{new-dense-fill}{$\enskip\ D_\y$}};
\else
\node at (dense.center) {\contourlength{0pt}\contour{new-dense-fill}{$D_\y$}};
\fi
\fi

\ifnum\m>0
\node[fit=(\matchB-3)(\matchBend-3), matchhatch, draw=new-match-line-end] (dense) {};
\node at (dense.center) {$D_\m$};
\fi
}

\draw[thick] (prefix.north east) to (suffix.north west);
\draw[thick] (prefix.south east) to (suffix.south west);

\node[left=0 of prefix] {$\phantom{T'}\mathllap{T'} =\enskip$};

\end{tikzpicture}

\vspace{.5\baselineskip}
}

\caption{Supplementary drawings for \cref{sec:finalimprov}.}
\label{fig:improv}
\end{figure}

For our final improvement we need to replace the LCE queries implemented by \cref{lem:LCE_undordered} with our own
mechanism.
The goal will remain the same, that is, given a parameter $\Delta$ and estimate $\tilde{\sigma}$ of the alphabet size,
find a $\Delta$-approximate LZ factorisation of any fragment $T[x..y]$ in $\Oh(m\tilde{\sigma}\log m/\sqrt{\Delta})$ time, where $m=|T[x..y]|$ (with $m=\Theta(\Delta^2)$, as otherwise we are not required to detect anything).

As in the previous section, the algorithm might detect that the size of the alphabet is larger
than $\tilde{\sigma}$, and in such case we revert to the divide-and-conquer algorithm. Let $\tau=\lfloor\sqrt{\Delta}\rfloor$.

Initially, we only consider some fragments of $T[x..y]$. We say that $T[i\cdot \tau^{2}..i\cdot \tau^{2} + \tau)$ is a dense fragment.
We start by remapping the characters in all dense fragments that intersect $T[x..y]$ to a linearly-sortable alphabet. This can be done
in $\Oh(\tilde{\sigma})$ time for each position by maintaining a list of the already seen distinct characters. For each position in a dense fragment,
we iterate over the characters in the list, and possibly append a new character to the list if it is not present. As soon as the size of
the list exceeds $\tilde{\sigma}$, we terminate the procedure and revert to the divide-and-conquer algorithm. Otherwise, we replace each character by its position in the list.
Overall, there are $\Oh(m/\sqrt{\Delta})$ positions in the dense fragments of $T[x..y]$, and the remapping takes $\Oh(m\tilde{\sigma}/\sqrt{\Delta})$ time.

Next, we construct two generalised suffix trees~\cite{Gusfield1997}, the first one of all dense fragments, and the second one of their reversals. 
(The generalised suffix tree of a collection of strings is the compacted trie that contains all suffixes of all strings in the collection.)
Again, because we
now work with a linearly-sortable alphabet this takes only $\Oh(m/\sqrt{\Delta})$ time~\cite{Farach1997}. We consider fragments
of the form $T[i\cdot \tau .. (i+1)\cdot \tau)$ having non-empty intersection with $T[x..y]$. We call such fragments chunks.
We note that there are $\Oh(m/\sqrt{\Delta})$ chunks, and their total length is $\Oh(m)$.
For each chunk, we find its longest prefix $T[i\cdot \tau.. i\cdot \tau+\ell_{i})$ and
longest suffix $T[(i+1)\cdot \tau-r_{i} .. (i+1)\cdot \tau)$ that occur in one of the dense fragments.
\cref{fig:improv:1} visualizes the dense fragments, chunks, and longest prefixes and suffixes.
This can be done efficiently by following the heavy path decomposition of the generalised suffix tree of all dense fragments
and their reversals, respectively. On each current heavy path, we just naively match the characters as long as possible.
In case
of a mismatch, we spend $\Oh(\tilde{\sigma})$ time to descend to the appropriate subtree, which happens at most $\Oh(\log m)$ times due to the heavy path decomposition.
After having found $\ell_{i}$ and $r_{i}$, we test square-freenes of $T[i\cdot \tau.. i\cdot \tau+\ell_{i})$
and $T[(i+1)\cdot \tau-r_{i} .. (i+1)\cdot \tau)$. Because they both occur in dense fragments,
and we have remapped the alphabet of all dense fragments, we can use \cref{lem:fasterclassical} to implement
this in $\Oh(\ell_{i}+r_{i})$ time.
Thus, the total time per chunk is thus $\Oh(\tilde{\sigma}\log m)$ plus $\Oh(\ell_{i}+r_{i})$. The former sums up to
$\Oh(m\tilde{\sigma}\log m/\sqrt{\Delta})$, and we will later show that the latter can be amortised by deactivating blocks on the lower levels.

The situation so far is that we have remapped the alphabet of all dense fragments to linearly-sortable, and for every chunk we know its longest
prefix and suffix that occur in one of the dense fragments. We concatenate 
all fragments of the form $T[i\cdot \tau-r_{i-1} .. i\cdot \tau+\ell_{i})$ (intersected with $T[x..y]$) while adding distinct separators
in between to form a new string $T'$. 
We stress that, because we have remapped the alphabet of all dense fragments, and the found
longest prefix and suffix of each chunk also occur in some dense fragment,  $T'$ is over linearly-sortable alphabet. 
Thus, we can build the suffix tree $ST$ of $T'$ in $\Oh(|T'|)$ time~\cite{Farach1997}. 
A visualization of $T'$ is provided in \cref{fig:improv:2}

Let $\mathcal{D}=\{D_{1},D_{2},\ldots\}$ be the set of distinct dense fragments.
We would like to construct the set of all occurrences of the strings in $\mathcal{D}$ in $T[x..y]$.
Using the suffix tree of $T'$ we can retrieve all occurrences
of every $D_{j}$ in $T'$. We observe that, because of how we have defined $T[i\cdot \tau.. i\cdot \tau+\ell_{i})$ and
$T[(i+1)\cdot \tau-r_{i} .. (i+1)\cdot \tau)$, this will in fact give us all occurrences of every $D_{j}$ in the original $T[x..y]$.
To implement this efficiently, we proceed as follows. First, for every $i$ we traverse $ST$ starting from its root to find the (explicit or implicit)
node corresponding to the dense fragment $T[i\cdot \tau^{2}.. i\cdot \tau^{2} + \tau)$.
This takes only $\Oh(m\tilde{\sigma}/\sqrt{\Delta})$ time.
Then, all leaves in every subtree rooted at such a node correspond to occurrences of some $D_{j}$, and can be
reported by traversing the subtree in time proportional to its size, so at most $\Oh(|T'|)$ in total.
Finally, remapping the occurrences back to $T[x..y]$ can be done in constant time per occurrence by precomputing,
for every position in $T'$, its corresponding position in $T[x..y$], which can be done in $\Oh(|T'|)$ time when
constructing $T'$. Thus, in $\Oh(|T'|)$ time, we obtain the set $S$ of starting positions of all occurrences
of the strings in $\mathcal{D}$. We summarize the properties of $S$ below.

\begin{proposition}
$S$ admits the following properties:
\begin{enumerate}
\item For every $i\in [x,y]$ such that $i = 0 \pmod {\tau^{2}}$, $i\in S$.
\item For every $i\in [x,y-\tau]$, $i\in S$ if and only if $T[i..i+\tau) \in \mathcal{D}$.
\item $|S| \leq |T'|$.
\end{enumerate}
\end{proposition}

We now define a parsing of $T[x..y]\$ $ based on $S$.
Let $i_{1}<i_{2}<\ldots i_{k}$ be all the positions in $S$, that is, $(i_{j},i_{j+1}) \cap S = \emptyset $ for every $j=1,2,\ldots,k-1$.
For every $j=1,2,\ldots,k-1$, we create the phrase $T[i_{j}.. i_{j+1}+\tau)$. We add the last phrase $T[i_{k}.. y]\$$. We stress that
consecutive phrases overlap by $\tau$ characters, and each phrase begins with a length-$\tau$ fragment starting at a position in $S$.
This, together with property 2 of $S$, implies the following property.

\begin{observation}
\label{obs:prefixfree}
The set of distinct phrases is prefix-free.
\end{observation}

We would like to construct the compacted trie $\Tphrase$ of all such phrases, so that (in particular) we identify identical
phrases. 
We first notice that each phrase begins with a fragment $T[i_{j}..i_{j}+\tau)$ that has its corresponding
occurrence in $T'$. We note that, given a set of positions $P$ in $T$, we can find their corresponding
positions in $T'$ (if they exist) by sorting and scanning in $\Oh(|P|+|T'|)$ time.

Thus, we can
assume that for each $i_{j}$ we know its corresponding position $i'_{j}$ in $T'$.
Next, for each node of $ST$ we precompute its unique ancestor at string depth $\tau$ in $\Oh(|T'|)$ time.
Then, for every fragment $T[i_{j}..i_{j}+\tau)$ we can access its corresponding (implicit or explicit)
node of $ST$. This allows us to partition all phrases according to their prefixes of length $\tau$.
In fact, this gives us the top part of $\Tphrase$ containing all such prefixes in $\Oh(m/\sqrt{\Delta})$ time,
and for each phrase we can assume that we know the node of $\Tphrase$ corresponding to its length-$\tau$
prefix. 

To build the remaining part of $\Tphrase$, we partition the phrases into short and long.
$T[i_{j}.. i_{j+1}+\tau)$ is short when $i_{j+1} \leq i_{j}+\tau$ (meaning that its length is at most $2\tau$), and long otherwise.

We begin with constructing the compacted trie $\Tphrase'$ of all short phrases.
This can be done similarly to constructing the top part of $\Tphrase$, except that now
the fragments have possibly different lengths. However, every short phrase $T[i_{j}..i_{j+1}+\tau)$
occurs in $T'$ as $T'[i'_{j}..i'_{j+1}+\tau)$. We claim that the nodes of $ST$ corresponding
to every $T'[i'_{j}..i'_{j+1}+\tau)$ can be found in $\Oh(|T'|)$ time. This can be done
by traversing $ST$ in the depth-first order while maintaining a stack of all explicit nodes with string depth at least $\tau$
on the current path. Then, when visiting the leaf
corresponding to the suffix of $T'$ starting at position $i'_{j}$, we iterate over the current
stack to find the sought node. This takes at most $\Oh(|T'[i_{j}+\tau..i_{j+1}+\tau]|)$ time,
which sums up to $\Oh(|T'|)$. Having found the node of $ST$ corresponding to $T[i_{j}..i_{j+1}+\tau)$,
we extract $\Tphrase'$ from $ST$ in $\Oh(|T'|)$ time.

With $\Tphrase'$ in hand, we construct the whole $\Tphrase$ as follows. We begin with taking the union of $\Tphrase'$ and
the already obtained top part of $\Tphrase$, this can be obtained in $\Oh(|T'|)$ time. For each long
phrase $T[i_{j}.. i_{j+1}+\tau)$, we know the node corresponding to $T[i_{j}..i_{j}+\tau)$
and would like to insert the whole string $T[i_{j}.. i_{j+1}+\tau)$ into $\Tphrase$.
We perform the insertions in increasing order of $i_j$ (this will be crucial for amortising the time later).
This is implemented with a dynamic heavy path decomposition similarly as in \cref{sec:suffixtree},
however with one important change. Namely, we fix a heavy path decomposition of the
part of $\Tphrase$ corresponding to the union of $\Tphrase'$ and the top part of $\Tphrase$, and maintain a dynamic heavy
path decomposition of every subtree hanging off from this part. Thanks to this change,
the time to maintain the dynamic trie and all heavy path decompositions is $\Oh(m\log m/\sqrt{\Delta})$,
as there are only $\Oh(m/\sqrt{\Delta})$ long phrases. Next, for each long phrase
$T[i_{j}..i_{j + 1}+\tau)$, we begin the insertion at the already known node corresponding to $T[i_{j}..i_j+\tau)$,
and continue the insertion by following the heavy paths, first in the static heavy path decomposition
in the part of $\Tphrase$ corresponding to $\Tphrase'$, second in the dynamic heavy path
decomposition in the appropriate subtree.
On each heavy path, we naively match the characters as long as possible.
The time to insert a single phrase $T[i_{j}.. i_{j+1}+\tau)$ is $\Oh(\log m)$ (twice)
plus the length of the longest prefix
of $T[i_{j}+\tau.. i_{j+1}+\tau)$ equal to a prefix of $T[i_{j'}+\tau.. i_{j'+1}+\tau)$, for some $j'<j$.
The former sums up to another $\Oh(m\log m/\sqrt{\Delta})$, and we will
later show that the latter can be amortised by deactivating blocks on the lower levels.

$\Tphrase$ allows us to form metacharacters corresponding to the phrases, and transform $T[x..y]$ into a string $\Tparse$
of length $\Oh(|T'|)$ consisting of these metacharacters.
We build a suffix tree $\Sparse$ over this
string over linearly-sortable metacharacters in $\Oh(|T'|)$ time. Next, we convert it into the sparse suffix tree
$\Sparse'$ of all suffixes $T[i_{j}..y]$ as follows. Consider an explicit node $u\in\Sparse$ with children $v_{1},v_{2},\ldots,v_{d}$,
$d\geq 2$. We first compute the subtree $\mathcal{T}_{u}$ of $\Tphrase$ induced by the leaves corresponding to the first metacharacters
on the edges $(u,v_{i})$, for $i=1,2,\ldots,d$, and connect every $v_{i}$ to the appropriate leaf of $\mathcal{T}_{u}$.
This can be implemented in $\Oh(d)$ time, assuming constant-time lowest common ancestor queries on $\Tphrase$~\cite{BenderF00}
and processing the leaves from left to right with a stack, similarly as in the Cartesian tree construction algorithm~\cite{Vuillemin80}.
We note that the order on the leaves is the same as the order on the metacharacters, and hence no extra sorting is necessary.
Overall, this sums up to $\Oh(|T'|)$.
Next, we observe that, unless $u$ is the root of $\Sparse$, all metacharacters on the edges $(u,v_{i})$ correspond
to strings starting with the same prefix of length $\tau$. We obtain the subtree $\mathcal{T}_{u}'$ by truncating
this prefix (or taking $\mathcal{T}_{u}$ if $u$ is the root). Finally, we identify the root of $\mathcal{T}_{u}'$
with $u$, and every child $v_{i}$ with its corresponding leaf of $\mathcal{T}_{u}'$.
Because we truncate the overlapping prefixes of length $\tau$, after this procedure is executed on every node
of $\Sparse$ we obtain a tree $\Sparse'$ with the property that each leaf corresponds to a suffix
$T[i_{j}..y]$. Also, by \cref{obs:prefixfree}, the edges outgoing from every node start with different characters
as required.

By following an argument from the proof of \cref{lem:lz-log-star},
$\Sparse'$ allows us to determine, for every suffix $T[i_{j}..y]$, its longest prefix equal to a prefix of some $T[i'..y]$
with $i' < i_{j}$, as long as its length is at least $\tau$. Indeed, in such case we must have $i'\in S$ by property 2,
so in fact $i'=i_{j'}$ and it is enough to maximise the length of the common prefix
with all earlier positions in $S$, which can be done using $\Sparse'$. Thus, we either know that the length of this longest
prefix is less than $\tau$, or know its exact value (and the corresponding position $i'\in S$).

\begin{lemma}
\label{lem:synchro-facto}
For any parameter $\Delta \in [1,m]$ and estimate $\tilde{\sigma}$ of the alphabet size, a $(\Delta+\tau)$-approximate LZ factorisation of any
fragment $T[x..y]$ can be computed in $\Oh(m/\sqrt{\Delta})$ time with $m=|T[x..y]|$ (assuming the preprocessing
described earlier in this section).
\end{lemma}

\begin{proof}
Let $e \in [x,y]$ and suppose we have already constructed the factorisation of $T[x..e-1]$ and are now trying to construct the next phrase.
Let $e'$ be the next multiple of $\tau^{2}$, we have that $e'-e < \tau^{2}\leq \Delta$ and $T[e'.. e'+\tau)$ is a dense fragment.
Thus, by property 1 we have $e'\in S$.

The first possibility is that the longest common prefix between $T[e'..y]$ and any suffix starting at an earlier position is shorter
than $\tau$. In this case, we can simply set the head of the new phrase to be $T[e..e'+\tau)$ and the tail to be empty.
Otherwise, we know the length $\ell$ of this longest prefix by the preprocessing described above.
We set the head of the new phrase to be $T[e..e')$ and the tail to be $T[e'..e''+\ell)$. This takes constant time per phrase, and each
phrase is of length at least $\tau$, giving the claimed overall time complexity. It remains to argue correctness of every step.

Let $T[e..s]$ be the longest LZ phrase starting at position $e$, to show that we obtain a valid $(\Delta+\tau)$-approximate 
phrase it suffices to show that $s \leq  e'+ \max(\tau,\ell)$.
Let the previous occurrence of $T[e..s)$ be at position $p<e$. If $s-e' < \tau$ then there is nothing to prove.
Otherwise, $T[e'..s)$ is a string of length at least $\tau$ that also occurs starting earlier at position $p+e'-e < e'$.
Thus, we will correctly determine that $\ell \geq \tau$, and find a previous occurrence of the string maximising
the value of $\ell$. In particular, we will have $\ell \geq s-e'$ as required.
\end{proof}

To achieve the bound of \cref{thm:upperbound2}, we now proceed as in \cref{sec:improved}, except that instead
of \cref{lem:compute3} we use \cref{lem:synchro-facto}. For every $T[x..y]$ with $m=|T[x..y]|$
this takes $\Oh(m\tilde{\sigma}\log m/\sqrt{\Delta})$ time plus the time used for computing the longest
prefix and suffix of each chunk (the latter also accounts for constructing the suffix tree $ST$ and other steps that have
been estimated as taking $\Oh(|T'|)$ in the above reasoning)
plus the time for inserting $T[i_{j}+\tau.. i_{j+1}+\tau)$ into $\Tphrase$ when $i_{j+1}\geq i_{j}+\tau$.

We observe that we can deactivate any block pair fully contained in $T[i\cdot \tau.. i\cdot \tau+\ell_{i})$ and $T[(i+1)\cdot \tau-r_{i} .. (i+1)\cdot \tau)$,
as we have already checked that these fragments are square-free.
Also, we can deactivate any block pair fully contained in 
the longest prefix of $T[i_{j}+\tau.. i_{j+1}+\tau)$ equal to $T[i_{j'}+\tau.. i_{j'+1}+\tau)$, for some $j'<j$,
because such fragment cannot contain the leftmost occurrence of a square. 

There are $\Oh(m / \sqrt{\Delta})$ chunks and long phrases. If a chunk or a long phrase
contributes $x = \Omega(\sqrt[4]{\Delta})$ to the total time, then we explicitly deactivate the block pairs in phase $t + 3$
that are entirely contained in the corresponding fragment. Block pairs in phase $t + 3$ are of length $\Oh(\sqrt[4]{\Delta})$,
and thus we deactivate $\Omega(x)$ positions.
Therefore, the time spent on such chunks and long phrases in all phases sums to $\Oh(n)$.
The remaining chunks and long phrases contribute $\Oh(\sqrt[4]{\Delta})$ to the total time,
and there are $\Oh(m / \sqrt{\Delta})$ of them, which adds up to $\Oh(m / \sqrt[4]{\Delta})$.
In every phase, this is $\Oh(n/\sqrt[4]{\Delta})$, so $\Oh(n)$ overall by \cref{lem:polylog}.

\section{Computing Runs}\label{sec:runs}

Now we adapt the algorithm such that it computes all runs. We start with the algorithm from \cref{sec:upper,sec:alg} without the final improvement from \cref{sec:finalimprov}. First, note that the key properties of the $\Delta$-approximate LZ factorisation, in particular \cref{lem:longhelper,lem:long}, also hold for the computation of runs. This is expressed by the lemmas below.

\begin{lemma}
\label{lem:longhelper:run}
Let $b_{1}b_{2}\ldots b_{z}$ be a $\Delta$-approximate LZ factorisation of a string $T$. For every run $\tuple{s,e,p}$ of length $e-s+1 \geq 8\Delta$, there is at least one phrase $b_i$ with $\absolute{\tail(b_i)} \geq \frac {e-s+1} 8 \geq \Delta$ such that $\tail(b_i)$ and the right-hand side $T[s + \ceil{\frac{e - s + 1} 2}..e]$ of the run intersect.
\end{lemma}
\begin{proof}
Let $\ell = \frac {e-s+1} 2$ and note that $\frac \ell 4 \geq \Delta$ and $e = s+2\ell - 1$. Assume that all tails that intersect $T[s + \ceil{\ell}..e]$ are of length less than $\frac \ell 4$, then the respective phrases of these tails are of length at most $\frac \ell 4 + \Delta - 1 \leq \frac \ell 2 - 1$ (because each head is of length less than $\Delta$). 
This means that $T[s + \ceil{\ell}..e]$ (of length $\floor{\ell}$) intersects at least $\ceil{\floor{\ell} / (\frac \ell 2 - 1)} \geq 3$ phrases (the inequality holds for $\ell \geq 4$, which is implied by $\Delta \geq 1$).
Thus there is some phrase $b_i = T[x..y]$ properly contained in $T[s+\ceil{\ell}..e]$, formally $s+\ceil{\ell} < x \leq y < e$.
However, this contradicts the definition of the $\Delta$-approximate LZ factorisation because $T[x..e+1]$ is the prefix of a standard LZ phrase (due to $T[x..e] = T[x - p..e-p]$). 
The contradiction implies that $T[s+\ceil{\ell}..e]$ intersects a tail of length at least $\frac \ell 4$.
\end{proof}

\noindent Before we show how to algorithmically apply \cref{lem:longhelper:run}, we need to explain how \cref{lem:conquer} extends
to computing runs, and then how this implies that the approach of Main and Lorentz~\cite{Main1984} easily extends to computing all runs.
We do not claim this to be a new result, but the original paper only talks about finding a representation of all squares, and
we need to find runs, and hence include a description for completeness.

\begin{lemma}
\label{lem:conquerruns}
Given two strings $x$ and $y$ over a general alphabet, we can compute all runs in $xy$ that include either the last character
of $x$ or the first character of $y$ using $\Oh(\absolute{x}+\absolute{y})$ time and comparisons.
\end{lemma}

\begin{proof}
Consider a run $\tuple{s,e,p}$ in $t=xy$ that includes either the last character of $x$ or the first character of $y$,
meaning that $s\leq \absolute{x}+1$ and $e\geq \absolute{x}$.
Let $\ell = \lfloor\frac {e-s+1} 2\rfloor \geq p$. We separately compute all runs with $s+\ell \leq \absolute{x}+1$ and $s+\ell > \absolute{x}+1$. Below we describe
the former, and the latter is symmetric.

Due to $s+\ell \leq \absolute{x}+1$, the length-$p$ substring $x[\absolute{x}-p+1..\absolute{x}]$ is fully
within the run. This suggests the following strategy to generate all runs with $s+\ell \leq \absolute{x}+1$. We iterate over the possible values of $p=1,2,\dots, \absolute{x}$. For a
given $p$, we calculate the length of the longest common prefix of $x[\absolute{x}-p+1..\absolute{x}]y$ and $y$, denoted $\textsf{pref}$,
and the length of the longest common suffix of $x[1..\absolute{x}-p]$ and $x$, denoted $\textsf{suf}$. 
It is easy to see that $t[\absolute{x}-p+1-\textsf{suf}..\absolute{x} + \textsf{pref}]$ is a lengthwise maximal $p$-periodic substring, and its length is $\ell' = p + \textsf{suf} + \textsf{pref}$.
If $\textsf{pref}+\textsf{suf} \geq p$ and $s+\floor{{\ell'}/2} \leq \absolute{x}+1$, then we report the substring as a run. (The latter condition ensures that each run gets reported by exactly one of the two symmetric cases.)

We use a prefix table to compute the longest common prefixes. For a given string, this table contains at position $i$ the length of the longest substring starting at position $i$ that is also a prefix of the string.
For computing the values $\textsf{pref}$, we use the prefix table of $y \$ x y$ (where $\$$ is a new character that does not match any character in $x$ nor $y$).
Similarly, for computing the values $\textsf{suf}$, we use the prefix table of the reversal of a new string $x \$ x$.
The tables can be computed in $\Oh(\absolute{x}+ \absolute{y})$ time and comparisons (see, e.g., computation of table \textit{lppattern} in \cite{Main1984}).
Then, each value of $p$ can be checked in constant time.
\end{proof}

\begin{lemma}
\label{lem:divideruns}
Computing all runs in a length-$n$ string over a general unordered alphabet can be implemented in $\Oh(n\log n)$ time and comparisons.
\end{lemma}

\begin{proof}
Let the input string be $T[1..n]$. We apply divide-and-conquer. Let $x=T[1..\lfloor n/2\rfloor]$ and $y=T[\lfloor n/2\rfloor + 1..n]$.
First, we recursively compute all runs in $x$ and $y$. Of the reported runs, we filter out all the ones that contain either the last character of $x$ or the first character of $y$, which takes $\Oh(\absolute{x} + \absolute{y})$ time.
In this way, if some reported run is a run with respect to $x$ (or $y$), but not with respect to $xy$, then it will be filtered out.
We have generated all runs except for the ones that contain the last character of $x$ or the first character of $y$ (or both). Thus we simply invoke
\cref{lem:conquerruns} on $xy$, which will output exactly the missing runs in $\Oh(\absolute{x} + \absolute{y})$ time and comparisons.
There are $\Oh(\log n)$ levels of recursion, and each level takes $\Oh(n)$ time and comparisons in total.
\end{proof}

\begin{lemma}
\label{lem:long:run}
Let $T = b_{1}b_{2}\ldots b_{z}$ be a $\Delta$-approximate LZ factorisation, and $\chi = \sum_{\absolute{\tail\left(b_{i}\right)} \geq \Delta}\absolute{\tail\left(b_{i}\right)}$.
We can compute in $\Oh\left(\chi + z\right)$ time and $\Oh\left(\chi\right)$ comparisons
a multiset $R$ of size $\Oh(\chi)$ of runs with the property that a run $T[s..e]$ is possibly not in $R$ only if $e-s+1 < 8\Delta$ or
there is some tail $\tail(b_i) = T[a_2..a_3]$ with $a_2 < s$ and $e < a_3$.
\end{lemma}

\begin{proof}
\newcommand{\targettime}{\Sigma_{\textsf{target}}}
Let $n = \absolute{T}$. We consider each phrase $b_i = T[a_1..a_3]$ with $\head(b_i) = T[a_1..a_2 - 1]$ and $\tail(b_i) = T[a_2..a_3]$ separately. Let $k = \absolute{\tail(b_i)}$. 
If $k \geq \Delta$, we apply \cref{lem:conquerruns} to $x_1 = T[a_2 - 8k..a_2 - 1]$ and $y_1 = T[a_2..a_3 + 4k]$, as well as $x_2 = T[a_2 - 8k..a_3 - 1]$ and $y_2 = T[a_3..a_3 + 4k]$ trimmed to $T[1..n]$. This takes $\Oh(\absolute{\tail(b_i)})$ time and comparisons and reports $\Oh(\absolute{\tail(b_i)})$ runs with respect to $x_1y_1=x_2y_2=T[a_2 - 8k..a_3 + 4k]$ (trimmed to $T[1..n]$). 
Of these runs, we filter out the ones that contain any of the positions $a_2 - 8k$ (only if $a_2 - 8k > 1$) and $a_3 + 4k$ (only if $a_3 + 4k < n$), which takes $\Oh(\absolute{\tail(b_i)})$ time.
This way, each reported run is not only a run with respect to $x_1y_1$, but also a run with respect to~$T$.
In total, we report $\Oh(\chi)$ runs (including possible duplicates) and spend $\Oh\left(\chi\right)$ time and comparisons when applying \cref{lem:conquerruns}. Additional $\Oh(z)$ time is needed to check if $\absolute{\tail(b_i)} \geq \Delta$ for each phrase.

Now we show that the described strategy computes all runs of length at least $8\Delta$, except for the ones that are properly contained in a tail. Let $\tuple{s,e,p}$ be a run of length $2\ell$, where $\ell \geq 4\Delta$ is a multiple of $\frac 1 2$.
Due to \cref{lem:longhelper:run}, the right-hand side $T[s+\ceil{\ell}..e]$ of this run intersects some tail $\tail(b_i) = T[a_2..a_3]$ of length $k = \absolute{\tail(b_i)} \geq \frac \ell 4 \geq \Delta$.
Due to the intersection, we have $a_2 \leq e$ and $a_3 \geq s + \ceil{\ell}$.
Thus, when processing $b_i$ and applying \cref{lem:conquerruns}, the starting position of $x_1$ and $x_2$ satisfies $a_2 - 8k \leq e - 8\frac \ell 4 < s$, while the end position of $y_1$ and $y_2$ satisfies $a_3 + 4k \geq s + \ceil{\ell} + 4\frac \ell 4 > e$. 
Therefore, the run is contained in the fragment $T[a_2 - 8k..a_3 + 4k]$ (trimmed to $T[1..n]$) corresponding to $x_1y_1$ and $x_2y_2$, and the run does not contain positions $a_2 - 8k$ and $a_3 + 4k$.
If $s \leq a_2 \leq e$, we find the run when applying \cref{lem:conquerruns} to $x_1$ and $y_1$.
If $s \leq a_3 \leq e$, we find the run when applying \cref{lem:conquerruns} to $x_2$ and $y_2$.
Otherwise, $T[s..e]$ is entirely contained in $T[a_2 + 1..a_3 - 1]$ and we do not have to report the run.
\end{proof}

Now we describe how to compute all runs using $\Oh(n \log \sigma)$ comparisons and $\Oh(n \log \sigma + n \log^* n)$ time.
We again use the sequence $\sigma_t=2^{2^{\lceil\log\log n\rceil-t}}$, for $t=0,1,\ldots,\lceil\log\log n\rceil$.
We observe that $\sigma_{t-1}=(\sigma_{t})^{2}$, and proceed in phases corresponding to the values of $t$.
In the $t^{\text{th}}$ phase we aim to compute runs of length at least $\sigma_t$ and less than $(\sigma_t)^2$.
We stress that this condition depends on the length of the run and not on its period.
We partition
the whole $T[1..n]$ into blocks of length $(\sigma_{t})^{2}$, and denote the $k^{\text{th}}$ block by $B_{k}$.
A run of length less than $(\sigma_{t})^{2}$ is fully contained within some two consecutive blocks $B_{i}B_{i+1}$, and there is always a pair of consecutive blocks such that the run contains neither the first nor the last position of the pair (unless the first position is $T[1]$ or the last position is $T[n]$ respectively).
Hence we consider each pair $B_{1}B_{2}$, $B_{2}B_{3}$, and so on.
We first apply \cref{lem:compute3} with $\Delta=\sigma_{t}/8$ and $\tilde{\sigma}=(\sigma_{t})^{1/4}/\log(\sigma_{t})$ to find
an $(\sigma_{t}/8)$-approximate LZ factorisation of the corresponding fragment of $T[1..n]$, and then
use \cref{lem:long:run} to compute all runs of length at least $\sigma_t$, apart from possibly the ones that are properly contained in a tail. Of the computed runs, we discard the ones that contain the first or last position of the block pair (unless the first position is $T[1]$ or the last position is $T[n]$ respectively).
This way, each reported run is a run not only with respect to the block pair, but with respect to the entire $T[1..n]$. If we do not report some run of length at least $\sigma_t$ and less than $(\sigma_t)^2$ in this way, then it is properly contained in one of the tails.

We cannot always afford to apply \cref{lem:long:run,lem:compute3} to all block pairs. Thus, we have to deactivate some of the blocks. During the current phase $t$, for each tail $T[s..e]$ of length at least $\Delta$, we deactivate all block pairs in phase $t + 3$ that are contained in $T[s + 1..e - 1]$. By similar logic as in \cref{sec:upper}, if a tail contributes $e - s + 1$ comparisons and time to the application of \cref{lem:long:run}, then it permanently deactivates $\Omega(e - s + 1)$ positions of the string, and thus the total time and comparisons needed for all invocations of \cref{lem:long:run,lem:compute3} are bounded by $\Oh(n)$ (apart from the additional $\Oh(n \log^* n)$ total time for \cref{lem:compute3}).
Whenever we apply \cref{lem:compute3}, we add all the tails of length at least $\Delta$ to a list $\mathcal L$, where each tail is annotated with the position of its previous occurrence.
After the algorithm terminates, $\mathcal L$ contains all sufficiently long tails from all phases. We have already shown that the total time needed for \cref{lem:long:run} is bounded by $\Oh(n)$, and thus the total length of the tails in $\mathcal L$ is at most $\Oh(n)$.

If any of the calls to \cref{lem:compute3} in the current phase detects that $\sigma>\tilde{\sigma}$, or if $\tilde{\sigma} < 256$, we immediately switch to
applying \cref{lem:divideruns} on every pair of blocks $B_{i}B_{i+1}$ of the current phase, which takes $\Oh(n \log \sigma)$ time (because the length of a block pair is polynomial in $\tilde{\sigma}$). Again, after applying \cref{lem:compute3} to $B_{i}B_{i+1}$, we discard all runs that contain the first or last position of $B_{i}B_{i+1}$ (unless the first position is $T[1]$ or the last position is $T[n]$, respectively).
After this procedure terminates, we have computed all runs, except for possibly some of the runs that were properly contained in a tail in list $\mathcal L$.
We may have reported some duplicate runs, which we filter out as follows.
The number of runs reported so far is $r=\Oh(n\log\sigma)$\footnote{a more careful analysis would reveal that it is $\Oh(n)$, but this is not necessary for the proof}.
We sort them in additional $\Oh(n+r)=\Oh(n\log\sigma)$ time, e.g., by using radix sort, and remove duplicates.
The running time so far is $\Oh(n\log \sigma)$.

\subsection{Copying Runs From Previous Occurrences}
\label{sec:copyruns}

Lastly, we have to compute the runs that were properly contained in a tail in $\mathcal L$. 
Consider such a run $\tuple{r_s, r_e, p}$, and let $T[s..e]$ be a tail in $\mathcal L$ with $s < r_s$ and $r_e < e$. If multiple tails match this criterion, let $T[s..e]$ be the one that maximizes $e$.
In $\mathcal L$, we annotated $T[s..e]$ with its previous occurrence $T[s-d..e-d]$.
Note that $\tuple{r_s - d, r_e - d, p}$ is also a run.
Thus, if we compute the runs in an appropriate order, we can simply copy the missing runs from their respective previous occurrences. For this sake, we annotate each position $i \in [1, n]$ with:

\begin{itemize}
\item a list of all the runs $\tuple{i, e, p}$ that we already computed, arranged in increasing order of end position $e$. We already sorted the runs for duplicate elimination, and can annotate all position in $\Oh(n)$ time.
\item a pair $(e^*, d^*)$, where $e^* = d^* = 0$ if there is no tail $T[s..e]$ such that $s < i < e$. Otherwise, among all tails $T[s..e]$ with $s < i < e$, we choose the one that maximizes $e$. Let $T[s-d..e-d]$ be its previous occurrence, then we use $e^* = e$ and $d^* = d$. As explained earlier, the total length of all tails in $\mathcal L$ is $\Oh(n)$, and thus we can simply scan each tail and update the annotation pair of each contained position whenever necessary.
\end{itemize}

Observe that, if a position is annotated with $(0,0)$, then none of the runs starting at position $i$ is fully contained in a tail, and thus we have already annotated position $i$ with the complete list of the runs starting at~$i$. 
Now we process the positions $i \in [1, n]$ one at a time and in increasing order.
We inductively assume that, at the time at which we process $i$, we have already annotated each $j < i$ with the complete list of runs starting at $j$. Hence our goal is to complete the list of $i$ such that it contains all runs starting at $i$. If $i$ is annotated with $(0,0)$, then the list is already complete.
Otherwise, $i$ is annotated with $(e, d)$, every missing run $\tuple{i, e_r, p}$ satisfies $e_r < e$, and the annotation list of $i - d$ already contains the run $\tuple{i - d, e_r - d, p}$ (due to $T[i - 1..e_r + 1] = T[i - d - 1..e_r - d + 1]$ and the inductive assumption).
For each run $\tuple{i - d, r_e - d, p}$ in the annotation list of position $i - d$, we insert the run $\tuple{i, e_r, p}$ into the annotation list of $i$. We perform this step in a merging fashion, starting with the shortest runs of both lists and zipping them together. As soon as we are about to insert a run $\tuple{i, e_r, p}$ with $e_r \geq e$, we do not insert it and abort. Thus, the time needed for processing $i$ is linear in the number of runs starting at position~$i$. By the runs theorem~\cite{Bannai2017}, the total number of runs is less than $n$, making the total time for this step $\Oh(n)$.

Apart from the new steps in \cref{sec:copyruns}, the complexity analysis works exactly like in \cref{sec:upper}. Hence we have shown:

\begin{theorem}
\label{thm:runsstar}
Computing all runs in a length-$n$ string that contains $\sigma$ distinct symbols from a general unordered alphabet can be implemented in $\Oh(n\log \sigma)$ comparisons and $\Oh(n\log \sigma + n \log^* n)$ time.
\end{theorem}

\subsection{Final Improvement for Computing Runs}

The goal is now to adapt the final algorithm to detect all runs. We can no longer stop as soon as we detect a square, and we cannot
simply deactivate pairs of blocks that occur earlier. However, \cref{lem:fasterclassical} is actually capable of reporting all runs in
$T[i\cdot \tau.. i\cdot \tau+\ell_{i})$ and $T[(i+1)\cdot \tau-r_{i} .. (i+1)\cdot \tau)$ in $\Oh(\ell_{i}+r_{i})$ time, and we
do not need to terminate the algorithm if these fragments are not square-free.
Thus, we can indeed deactivate any block pair fully contained in $T[i\cdot \tau.. i\cdot \tau+\ell_{i})$ and $T[(i+1)\cdot \tau-r_{i} .. (i+1)\cdot \tau)$.
Next, we also deactivate block pairs fully contained in the longest prefix of $T[i_{j}+\tau.. i_{j+1}+\tau)$ equal to
$T[i_{j'}+\tau.. i_{j'+1}+\tau)$, for some $j'<j$. Denoting the length of this prefix by $\ell$,
we treat $T[i_{j}+\tau.. i_{j}+\ell)$ as a tail and add it to the list $\mathcal{L}$ (annotated with $i_{j'}$).
The total length of all fragments added to $\mathcal{L}$ is still $\Oh(n)$.

\begin{theorem}
\label{thm:runs}
Computing all runs in a length-$n$ string that contains $\sigma$ distinct symbols from a general unordered alphabet can be implemented in $\Oh(n\log \sigma)$ comparisons and $\Oh(n\log \sigma)$ time.
\end{theorem}

\bibliographystyle{abbrv}
\bibliography{references}

\end{document}